\documentclass[10pt,a4paper,twocolumn]{article}
\usepackage[top=30pt,bottom=30pt,left=48pt,right=46pt]{geometry}

\usepackage{authblk}
\usepackage{amssymb,amsmath,amsthm}
\usepackage{mathtools}
\usepackage{graphicx}
\usepackage[latin1]{inputenc}
\usepackage[T1]{fontenc}
\usepackage{booktabs}
\usepackage{array}
\usepackage{siunitx}
\usepackage{multirow,makecell}
\usepackage{url}
\usepackage{float}
\usepackage{stmaryrd}
\usepackage{xfp}
\usepackage{varwidth}

\usepackage{mathrsfs}%
\usepackage[figuresright]{rotating}%
\usepackage[title]{appendix}%
\usepackage{xcolor}%
\usepackage{textcomp}%
\usepackage{manyfoot}%
\usepackage{booktabs}%
\usepackage{algorithm}%
\usepackage{algorithmicx}%
\usepackage{algpseudocode}%
\usepackage{program}%
\usepackage{listings}%

\newtheorem{definition}{Definition}[section]
\newtheorem{theorem}{Theorem}[section]
\newtheorem{corollary}{Corollary}[theorem]
\newtheorem{lemma}[theorem]{Lemma}

\usepackage{bookmark}
\usepackage{hyperref}

\hypersetup{
  colorlinks   = true, 
  urlcolor     = blue, 
  linkcolor    = blue, 
  citecolor   = blue 
}

\usepackage{arydshln}
\usepackage{xparse}
\makeatletter
\NewDocumentCommand{\LeftComment}{s m}{%
  \Statex \IfBooleanF{#1}{\hspace*{\ALG@thistlm}}\(\triangleright\) #2}
\makeatother

\renewcommand{\ge}{\geqslant}
\renewcommand{\geq}{\geqslant}
\newcommand{\ratioc}[2]{${#1}$ && ${#2}$ && $\times$\fpeval{trunc({#1} / {#2}, 2)}}

\newcommand{\fv}[1]{\ensuremath{\mathit{fv}({#1})}}
\newcommand{\succs}[1]{\ensuremath{\mathord{\downarrow}{#1}}}

\newcommand{\caesar}{\textsc{c{\ae}sar.bdd}}
\newcommand{\conc}{\ensuremath{\mathbin{\|}}}
\newcommand{\Conc}{\ensuremath{\mathrm{C}}}
\newcommand{\bconc}{\mathop{\mathcal{C}}}
\newcommand{\bind}{\mathop{\bar{\mathcal{C}}}}
\newcommand{\kong}{\textsc{Kong}}
\newcommand{\reduce}{\textsc{Reduce}}
\newcommand{\sift}{\textsc{Sift}}
\newcommand{\red}{\mathbin{{\rightarrow}\!{\bullet}}}
\newcommand{\agg}{\mathbin{{\circ}\!{\rightarrow}}}
\newcommand{\lvl}[1]{\mathrm{lvl}({#1})}
\newcommand{\maseq}[1]{\lfloor {#1} \rfloor}
\newcommand{\comma}{\mathbin{\raisebox{0.2ex}{\textbf{,}}}}
\newcommand{\Nat}{\mathbb{N}}

\newcommand{\reduc}{\vartriangleright}
\newcommand*{\defeq}{\triangleq} 

\newif\ifmarkingSep

\newcommand{\markingScan}[2]{%
  \ifx\relax#1\empty
  \else
    \ifmarkingSep
      {\,}\relax
    \else
      \markingSeptrue
    \fi
    {#1}{\ast}{#2}\relax
    \expandafter\markingScan
  \fi
}

\makeatletter
\def \rightarrowfill{\m@th\mathord{\smash-}\mkern-6mu%
  \cleaders\hbox{$\mkern-2mu\mathord{\smash-}\mkern-2mu$}\hfill
  \mkern-6mu\mathord\rightarrow}
\makeatother

\makeatletter
\def \Rightarrowfill{\m@th\mathord{\smash-}\mkern-6mu%
  \cleaders\hbox{$\mkern-2mu\mathord{\smash-}\mkern-2mu$}\hfill
  \mkern-6mu\mathord\Rightarrow}
\makeatother

\makeatletter
\def \rightarrowfill{\m@th\mathord{\smash-}\mkern-6mu%
  \cleaders\hbox{$\mkern-2mu\mathord{\smash-}\mkern-2mu$}\hfill
  \mkern-6mu\mathord\rightarrow}
\makeatother

\makeatletter
\def \Rightarrowfill{\m@th\mathord{\smash=}\mkern-6mu%
  \cleaders\hbox{$\mkern-2mu\mathord{\smash=}\mkern-2mu$}\hfill
  \mkern-6mu\mathord\Rightarrow}
\makeatother

\makeatletter
\def \midrightarrowfill{\m@th\mathord{\smash{\raisebox{.2ex}{$\scriptscriptstyle\mid$}}\!\!\,-}\mkern-6mu%
  \cleaders\hbox{$\mkern-2mu\mathord{\smash-}\mkern-2mu$}\hfill
  \mkern-6mu\mathord\rightarrow}
\makeatother

\makeatletter
\def \midRightarrowfill{\m@th\mathord{\smash{\raisebox{.1ex}{$\scriptstyle\mid$}}\!\!\!=}\mkern-6mu%
  \cleaders\hbox{$\mkern-2mu\mathord{\smash=}\mkern-2mu$}\hfill
  \mkern-6mu\mathord\Rightarrow}
\makeatother

\newcommand{\overstackrel}[2]{\mathrel{\mathop{#1}\limits^{#2}}}
\newcommand{\trans}[1]{\mathbin{\smash[t]{\overstackrel{\rightarrowfill}{\ #1\ }}}}

\newcommand{\tfg}[2][{}]{\llbracket {#2} \rrbracket_{#1}}
\newcommand{\acomment}[1]{\textcolor{violet}{\small\textit{#1}}}

\usepackage{subcaption}
\usepackage{balance}

\providecommand{\keywords}[1]
{
  \small	
  \textbf{\textit{Keywords---}} #1
}

\title{Leveraging polyhedral reductions for solving Petri net reachability problems}
\author[1]{Nicolas Amat}
\author[1]{Silvano Dal Zilio} 
\author[1]{Didier Le Botlan}
\affil[1]{LAAS-CNRS, Universit\'{e} de Toulouse, CNRS, Toulouse, France}
\date{}
\begin{document}

\maketitle
\sloppy

\abstract{%
  We propose a new method that takes advantage of structural
  reductions to accelerate the verification of reachability properties
  on Petri nets.
  Our approach relies on a state space abstraction, called polyhedral
  abstraction, which involves a combination
  between structural reductions and sets of linear arithmetic
  constraints between the marking of places.
  We propose a new data-structure, called a Token Flow Graph (TFG),
  that captures the particular structure of constraints occurring in
  polyhedral abstractions. We leverage TFGs to efficiently solve two
  reachability problems: first to check the reachability of a given
  marking; then to compute the concurrency relation of a net, that is
  all pairs of places that can be marked together in some reachable
  marking.
  Our algorithms are implemented in a tool, called Kong, that we
  evaluate on a large collection of models used during the 2020
  edition of the Model Checking Contest.  Our experiments show that
  the approach works well, even when a moderate amount of reductions
  applies.}


\maketitle
 

\section{Introduction}

We propose a new method that takes advantage of structural reductions
to accelerate the verification of reachability properties on Petri
nets. In a nutshell, we compute a reduced net $(N_2, m_2)$, from an
initial marked net $(N_1, m_1)$, and prove properties about the
initial net by exploring only the state space of the reduced one. A
difference with previous works on structural reductions,
e.g.~\cite{berthelot_transformations_1987}, is that our approach is
not tailored to a particular class of properties---such as the absence
of deadlocks---but could be applied to more general problems.

To demonstrate the versatility of this approach, we apply it to two specific
problems: first to check the \textit{reachability} of a given marking; then to
compute the \emph{concurrency relation} of a net, that is all pairs of places
that can be marked together in some reachable marking.

On the theoretical side, the correctness of our approach relies on a new state
space abstraction method, that we called \emph{polyhedral abstraction}
in~\cite{tacas,fi2022}, which involves a set of linear arithmetic constraints
between the marking of places in the initial and the reduced net. The idea is to
define relations of the form $(N_1, m_1) \reduc_E (N_2, m_2)$, where $E$ is a
system of linear equations that relates the possible markings of $N_1$ and
$N_2$. More precisely, the goal is to preserve enough information in $E$ so that
we can rebuild the reachable markings of $N_1$ knowing only those of $N_2$.

On the practical side, we derive polyhedral abstractions by computing structural
reductions from an initial net, incrementally. We say in this case that we
compute a \emph{polyhedral reduction}. While there are many examples of the
benefits of structural reductions when model-checking Petri nets, the use of an
equation system ($E$) for tracing back the effect of reductions is new, and we
are hopeful that this approach can be applied to other problems.

Our algorithms rely on a new data structure, called a \emph{Token Flow Graph}
(TFG) in~\cite{spin2021}, that captures the particular structure of constraints
occurring in the linear system $E$. We describe TFGs and show how to leverage
this data structure in order to {accelerate} the computation of solutions for
the two reachability problems we mentioned: (1) marking reachability and (2)
concurrency relation. We use the term acceleration to stress the ``multiplicative
effect'' of TFGs. Indeed, we propose a framework that, starting from a tool for
solving problem (1) or (2), provide an augmented version of this tool that takes
advantage of reductions. The augmented tool can compute the solution for an
initial instance, say on some net $N$, by solving it on a reduced version of
$N$, and then reconstructing a correct solution for the initial instance. In
each case, our approach takes the form of an ``inverse transform'' that relies
only on $E$ and that does not involve expensive pre-processing on the reduced
net.

For the marking reachability problem, we illustrate our approach by
augmenting the tool \sift, which is an explicit-state model-checker
for Petri nets that can check reachability properties on the fly. For
the concurrency relation, we augment the tool \caesar, part of the
CADP toolbox~\cite{pbhg2021,cadp}, that uses BDD techniques to explore
the state space of a net and find concurrent places.

We show that our approach can result in massive speed-ups since the
reduced net may have far fewer places than the initial one, and since
the number of places is often a predominant parameter in the
complexity of reachability problems.

\paragraph{Outline and contributions}
After describing some related works, we define the semantics of Petri
nets and the notion of concurrent places in
Sect.~\ref{sec:petri-nets}.
We define a simplified notion of ``reachability equivalence'' in
Sect.~\ref{sec:polyhedral}, that we call a \emph{polyhedral
  abstraction}. Section~\ref{sec:tfg} and \ref{sec:semantics} contain
our main contributions. We describe Token Flow Graphs (TFGs) in
Sect.~\ref{sec:tfg} and prove several results about them in
Sect.~\ref{sec:semantics}. These results allow us to reason about the
reachable places of a net by playing a token game on the nodes of a
TFG.  We use TFGs to define a decision procedure for the reachability
problem in Sect.~\ref{sec:reachability}. Next,
in Sect.~\ref{sec:algorithm}, we define a similar algorithm for finding
concurrent places and show how to adapt it to situations where we only
have partial knowledge of the residual concurrency relation.

Our approach has been implemented and computing experiments show that reductions
 are effective on a large set of models (Sect.~\ref{sec:experimental}). Our
 benchmark is built from an independently managed collection of Petri nets
 corresponding to the nets used during the 2020 edition of the Model Checking
 Contest~\cite{mcc2019}.  We observe that, even with a moderate amount of
 reductions (say we can remove $25\%$ 
of the places), we can compute complete results much faster with
reductions than without; often by several orders of magnitude. We also
show that we perform well with incomplete relations, where we are both
faster and more accurate.

Many results and definitions were already presented in
\cite{spin2021}. This extended version contains several
additions. First, we extend our method based on polyhedral reductions
and the TFG to the reachability problem, whereas~\cite{spin2021} was
only about computing the concurrency relation. For this second
problem, we give detailed proofs for all our results about TFGs that
justify the intimate connection between solutions of the linear system
$E$ and reachable markings of a net. Our paper also contains
definitions and proofs for new axioms that are useful when computing
incomplete concurrency relations; that is in the case where we only
have partial knowledge on the concurrent places. Finally, we provide
more experimental results about the performance of our tool.


\section{Related work}

We consider works related to the two problems addressed in this paper,
with a particular emphasis on concurrent places, which provides the
most original results. We also briefly discuss the use of structural
reductions in model-checking.

\subsection{Marking reachability}
Reachability for Petri nets is an important and difficult problem with
many practical applications. In this work, we consider the simple
problem of checking whether a given marking $m'_1$ is reachable by
firing a sequence of transitions in a net $N_1$, starting from an
initial marking $m_1$.

In a previous work~\cite{tacas,fi2022}, we used polyhedral abstraction and
symbolic model-checking to augment the verification of ``generalized''
reachability properties, in the sense that we check whether it is
possible to reach a marking that satisfies a property $\phi$ expressed
as a Boolean combination of linear constraints between places, such as
$(p_0 + p_1 = p_2 + 1) \wedge (p_0 \leqslant p_2)$ for example.

This more general problem corresponds to one of the examinations in the
Model Checking Contest (MCC)~\cite{mcc2019}, an annual competition of
model-checkers for Petri nets. Many optimization techniques are used
in this context: symbolic techniques, such as $k$-induction; standard
abstraction techniques used with Petri nets, like stubborn sets and
partial order reduction; the use of the ``state equation''; reduction
to integer linear programming problems; etc.

Assume that $(N_2, m_2)$ is the polyhedral reduction of $(N_1, m_1)$,
with the associated set of equations $E$. The main result
of~\cite{tacas,fi2022} is that it is possible to build a formula $\phi_E$
such that $\phi$ is reachable in $N_1$ if and only if $\phi_E$ is
reachable in $N_2$. This means that we can easily augment any
model-checker---if it is able to handle generalized reachability
properties---so as to benefit from structural reductions for free.

In this paper, we use TFGs to prove a stronger property for the marking
reachability problem (see Sect.~\ref{sec:reachability}), namely that,
given a target marking $m'_1$ for $N_1$, we are able to effectually
compute a marking $m'_2$ of $N_2$ such that $m'_1$ is reachable in
$N_1$ if and only if $m'_2$ is reachable in $N_2$. This can be more
efficient than our previous method, since the transformed property
$\phi_E$ can be quite complex in practice, event though the property
for marking reachability is a simple conjunction of equality
constraints. For instance, we can perform our experiments using only a
basic, explicit-state model-checker.

This application of polyhedral reductions, while not as original as
our results with the concurrency relation, highlights the fact that
TFGs provide an effective method to exploit reductions. It also bears
witness to the versatility of our approach.

\subsection{Concurrency relation}
The main result of our work is a new approach for computing the
\emph{concurrency relation} of a Petri net. This problem has practical
applications, for instance because of its use for decomposing a Petri
net into the product of concurrent
processes~\cite{janicki_automatic_2020,garavel_nested-unit_2019}. It
also provides an interesting example of safety property that nicely
extends the notion of \emph{dead places}; meaning places that can
never be reached in an execution.  These problems raise difficult
technical challenges and provide an opportunity to test and improve
new model-checking techniques~\cite{garavel2021proposal}.

Naturally, it is possible to compute the concurrency relation by
checking, individually, the reachability of each pair of places. But
this amounts to solving a quadratic number of coverability
properties---where the parameter is the number of places in the
net---and one would expect to find smarter solutions, even if it is
only for some specific cases. We are also interested in partial
solutions, where computing the whole state space is not feasible.

Several works address the problem of finding or characterizing the concurrent
places of a Petri net. This notion is mentioned under various names, such as
\emph{coexistency defined by markings}~\cite{janicki_nets_1984},
\emph{concurrency graph}~\cite{wisniewski_prototyping_2018} or \emph{concurrency
relation}~\cite{garavel_state_2004,kovalyov_concurrency_1992,kovalyov_polynomial_2000,semenov_combining_1995,wisniewski2019c}.
The main motivation is that the concurrency relation characterizes the
sub-parts, in a net, that can be simultaneously active. Therefore, it plays a
useful role when decomposing a net into a collection of independent components.
This is the case in~\cite{wisniewski2019c}, where the authors draw a connection
between concurrent places and the presence of ``sequential modules'' (state
machines). Another example is the decomposition of nets into unit-safe NUPNs
(Nested-Unit Petri Nets)~\cite{janicki_automatic_2020,garavel_nested-unit_2019},
for which the computation of the {concurrency relation} is one of the main
bottlenecks.

We know only a couple of tools that support the computation of the
concurrency relation. A recent tool is part of the Hippo
platform~\cite{wisniewski2019c}, available online. Our reference tool
in this paper is \caesar, from the CADP
toolbox~\cite{pbhg2021,cadp}. It supports the computation of a partial
relation and can output the ``concurrency matrix'' of a net using a
specific, compressed, textual format~\cite{garavel2021proposal}. We
adopt the same format since we use \caesar\ to compute the concurrency
relation on the residual net, $N_2$, and as a yardstick in our
benchmarks.

\subsection{Model-checking with reductions}
Concerning our use of structural reductions, our main result can be
interpreted as an example of \emph{reduction
  theorem}~\cite{lipton_reduction_1975}, that allows to deduce
properties of an initial model ($N_1$) from properties of a simpler,
coarser-grained version ($N_2$). But our notion of reduction is more
complex and corresponds to the one pioneered by
Berthelot~\cite{berthelot_transformations_1987}, but with the
addition of linear equations.

Several tools use reductions for checking reachability properties, but none
specializes in computing the concurrency relation. We can mention
\textsf{Tapaal}~\cite{bonneland2019stubborn}, an explicit-state model-checker
that combines partial-order reduction techniques and structural reductions or,
more recently, \textsf{ITS-Tools}~\cite{thierry-mieg_structural_2020}, which
combines several techniques, including structural reductions and the use of SAT
and SMT solvers.

In our work, we focus on reductions that preserve the reachable states
and use ``reduction equations'' to keep traceability information
between initial and reduced nets. Our work is part of a trilogy.

Our approach was first used for \emph{model
  counting}~\cite{berthomieu2018petri,berthomieu_counting_2019}, as a
way to efficiently compute the number of reachable states. It was
implemented in a symbolic model-checker called \textsf{Tedd}, which is
part of the Tina toolbox~\cite{tinaToolbox}.
It also relies on an ancillary tool, called \reduce, that applies
structural reductions and returns the set of reduction equations. We
reuse this tool in our experiments.

The second part of our trilogy~\cite{tacas,fi2022} defines a method for
taking advantage of net reductions in combination with a SMT-based
model-checker and led to a new dedicated tool, called
\textsf{SMPT}. This work introduced the notion of polyhedral
abstraction. The main goal here was to provide a formal framework for
our approach, in the form of a new semantic equivalence between Petri
nets.

Finally, this paper is an extended version of~\cite{spin2021}, that
introduces the notion of Token Flow Graph and describes a new
application, to accelerate the computation of concurrent places. Our
goal here is to provide effective algorithms that can leverage the
notion of polyhedral abstraction. It is, in some sense, the practical
or algorithmic counterpart of the theory developed
in~\cite{tacas,fi2022}. This work is also associated with a new tool, called
\kong, that we describe in Sect.~\ref{sec:experimental}.


\section{Petri nets}
\label{sec:petri-nets}

Some familiarity with Petri nets is assumed from the reader. We recall some
basic terminology. Throughout the text, comparison $(=, \geq)$ and arithmetic
operations $(-, +)$ are extended pointwise to functions and tuples.

\begin{definition}[Petri net]
A \textit{Petri net} $N$ is a tuple $(P, T, \mathrm{Pre}, \mathrm{Post})$ where:
\begin{itemize}
  \item $P = \{p_1, \dots, p_n\}$ is a finite set of places,
  \item $T = \{t_1, \dots, t_k\}$ is a finite set of transitions (disjoint from
  $P$),
  \item $\mathrm{Pre} : T \rightarrow (P \rightarrow \mathbb{N})$ and
$\mathrm{Post} : T \rightarrow (P \rightarrow \mathbb{N})$ are the pre- and
post-condition functions (also called the flow functions of $N$).
\end{itemize}
\end{definition}

We often simply write that $p$ is a place of $N$ when $p \in P$.
A state $m$ of a net, also called a \emph{marking}, is a total mapping $m : P
\rightarrow \mathbb{N}$ which assigns a number of \emph{tokens}, $m(p)$, to each
place of $N$. A marked net $(N, m_0)$ is a pair composed of a net and its
initial marking $m_0$.

A transition $t \in T$ is \textit{enabled} at marking $m \in \Nat^P$ when $m(p)
\ge \mathrm{Pre}(t,p)$ for all places $p$ in $P$. (We can also simply write $m
\geq \mathrm{Pre}(t)$, where $\geq$ stands for the component-wise comparison of
markings.) A marking $m'$ is reachable from a marking $m$ by firing transition
$t$, denoted $m \trans{t} m'$, if: (1) transition $t$ is enabled at $m$; and (2)
$m' = m - \mathrm{Pre}(t) + \mathrm{Post}(t)$. When the identity of the
transition is unimportant, we simply write this relation $m \trans{} m'$.  More
generally, marking $m'$ is reachable from $m$ in $N$, denoted $m \trans{}^\star
m'$ if there is a (possibly empty) sequence of reductions such that $m \trans{}
\dots \trans{} m'$.  We denote $R(N, m_0)$ the set of markings reachable from
$m_0$ in $N$.

A marking $m$ is $k$-{bounded} when each place has at most $k$ tokens and a
marked Petri net $(N, m_0)$ is bounded when there is a constant $k$ such that
all reachable markings are $k$-bounded. While most of our results are valid in
the general case---with nets that are not necessarily bounded and without any
restrictions on the flow functions (the weights of the arcs)---our tool and our
experiments on the concurrency relation focus on the class of $1$-bounded nets,
also called \emph{safe} nets.

\begin{figure}[htbp]
  \centering
  \includegraphics[width=0.75\linewidth]{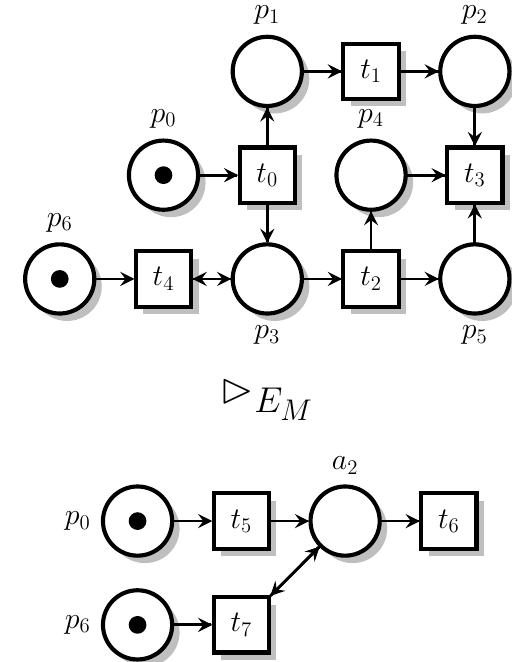}
  \caption{An example of Petri net, $M_1$ (top), and one of its polyhedral
    abstraction, $M_2$ (bottom), with $E_M \defeq (p_5 = p_4), (a_1 = p_1 + p_2),
    (a_2 = p_3 + p_4), (a_1 = a_2)$.\label{fig:stahl}}
\end{figure}

Given a marked net $(N, m_0)$, we say that places $p, q$ of $N$ are concurrent
when there exists a reachable marking $m$ with both $p$ and $q$ marked.
The \textit{concurrent places} problem consists in enumerating all such pairs of
places.

\begin{definition}[Dead and concurrent places]
  We say that a place $p$ of $(N, m_0)$ is \emph{nondead} if there is $m$ in
  $R(N,m_0)$ such that $m(p) > 0$. Similarly, we say that places $p,q$
  are \emph{concurrent}, denoted $p \conc q$, if there is $m$ in $R(N,m_0)$ such
  that both $m(p) > 0$ and $m(q) > 0$. By extension, we use the notation $p
  \conc p$ when $p$ is nondead. We say that $p, q$ are nonconcurrent, denoted
  $p \mathbin{\#} q$, when they are not concurrent.
\end{definition}

\paragraph{Relation with linear arithmetic constraints}
Many results in Petri net theory are based on a relation with linear algebra and
linear programming techniques~\cite{murata1989petri,silva1996linear}. A
celebrated example is that the potentially reachable markings (an
over-approximation of the reachable markings) of a net $(N, m_0)$ are
non-negative, integer solutions to the \emph{state equation} problem, $m = I
\cdot \sigma + m_0$, with $I$ an integer matrix defined from the flow functions
of $N$ called the incidence matrix and $\sigma$ a vector in $\Nat^k$.
It is known that solutions to the system of linear equations $\sigma^T \cdot I =
\vec{0}$ lead to \emph{place invariants}, $\sigma^T \cdot m = \sigma^T \cdot
m_0$, that can provide some information on the decomposition of a net into
blocks of nonconcurrent places, and therefore information on the concurrency
relation. 

For example, for net $M_1$ (Fig.~\ref{fig:stahl}), we can compute invariant $p_4
- p_5 = 0$. This is enough to prove that places $p_4$ and $p_5$ are concurrent,
if we can prove that at least one of them is nondead. Likewise, an invariant of
the form $p + q = 1$ is enough to prove that $p$ and $q$ are $1$-bounded and
cannot be concurrent. Unfortunately, invariants provide only an
over-approximation of the set of reachable markings, and it may be difficult to
find whether a net is part of the few known classes where the set of reachable
markings equals the set of potentially reachable ones~\cite{hujsa:hal-02992521}.

Our approach shares some similarities with this kind of reasoning. A main
difference is that we will use equation systems to draw a relation between the
reachable markings of two nets; not to express constraints about (potentially)
reachable markings inside one net. Like with invariants, this will allow us, in
many cases, to retrieve information about the concurrency relation without
``firing any transition'', that is without exploring the state space.

In the following, we will often use place names as variables, and markings $m :
P \to \Nat$ as partial solutions to a set of linear equations. For the sake of
simplicity, all our equations will be of the form $x = y_1 + \dots + y_l$ or
$y_1 + \dots + y_l = k$ (with $k$ a constant in $\Nat$).

Given a system of linear equations $E$, we denote $\fv{E}$ the set of all its
variables. We are only interested in the non-negative integer solutions of $E$.
Hence, in our case, a \emph{solution} to $E$ is a total mapping from variables
in $\fv{E}$ to $\Nat$ such that all the equations in $E$ are satisfied.  We say
that $E$ is \emph{consistent} when there is at least one such solution.
Given these definitions, we say that the mapping $m: \{p_1, \dots, p_n\} \to
\Nat$ is a (partial) solution of $E$ if the system $E \comma \maseq{m}$ is
consistent, where $\maseq{m}$ is the sequence of equations $p_1 = m(p_1) \comma
\dots\comma p_n = m(p_n)$. (In some sense, we use $\maseq{m}$ as a
substitution.) For instance, places $p, q$ are concurrent if the system $p = 1 +
x \comma q = 1 + y \comma \maseq{m}$ is consistent, where $m$ is a reachable
marking and $x, y$ are some fresh (slack) variables.

Given two markings $m_1 : P_1 \to \Nat$ and $m_2 : P_2 \to \Nat$, from possibly
different nets, we say that $m_1$ and $m_2$ are \emph{compatible}, denoted $m_1
\equiv m_2$, if they have equal marking on their shared places: $m_1(p) =
m_2(p)$ for all $p$ in $P_1 \cap P_2$. This is a necessary and sufficient
condition for the system $\maseq{m_1} \comma \maseq{m_2}$ to be consistent.

\section{Polyhedral abstraction}
\label{sec:polyhedral}

We recently defined a notion of \emph{polyhedral abstraction} \cite{tacas,fi2022} based
on our previous work applying structural reductions to model
counting~\cite{berthomieu2018petri,berthomieu_counting_2019}. We only need a simplified
version of this notion here, which entails an equivalence between the state
space of two nets, $(N_1, m_1)$ and $(N_2, m_2)$, ``up-to'' a system $E$ of
linear equations.

\begin{definition}[$E$-equivalence]
  \label{def:e-equivalence}
  We say that $(N_1, m_1)$ is $E$-equivalent to $(N_2, m_2)$, denoted
  $(N_1, m_1) \reduc_E (N_2, m_2)$, if and only if:
  \begin{description}
  \item [\textbf{(A1)}] $E \comma \maseq{m}$ is consistent for all markings $m$
    in $R(N_1, m_1)$ or $R(N_2, m_2)$;

  \item[\textbf{(A2)}] initial markings are \emph{compatible}, meaning
    $E \comma \maseq{m_1} \comma \maseq{m_2}$ is consistent;

  \item[\textbf{(A3)}] assume $m'_1, m'_2$ are markings of $N_1,N_2$,
    respectively, such that $E\comma \maseq{m'_1}\comma \maseq{m'_2}$ is
    consistent, then $m'_1$ is reachable if and only if $m'_2$ is
    reachable:\par $m'_1 \in R(N_1,m_1) \iff m_2' \in R(N_2, m_2)$.
  \end{description}
\end{definition}

By definition, relation $\reduc_E$ is symmetric. We deliberately use a symbol
oriented from left to right to stress the fact that $N_2$ should be a reduced
version of $N_1$. In particular, we expect to have fewer places in $N_2$ than in
$N_1$.

Given a relation $(N_1, m_1) \reduc_E (N_2, m_2)$, each marking $m'_2$ reachable
in $N_2$ can be associated to a unique subset of markings reachable in $N_1$, defined from
the solutions to $E\comma \maseq{m'_2}$ (by conditions (A1) and (A3)). We can
show~\cite{tacas,fi2022} that this gives a partition of the reachable markings of
$(N_1, m_1)$ into ``convex sets''---hence the name polyhedral abstraction---each
associated to a reachable marking in $N_2$. Our approach is particularly useful
when the state space of $N_2$ is very small compared to the one of $N_1$. In the
extreme case, we can even find examples where $N_2$ is the ``empty'' net (a net
with zero places, and therefore a unique marking), but this condition is not a
requisite in our approach.

We can illustrate this result using the two marked nets $M_1, M_2$ in
Fig.~\ref{fig:stahl}, for which we can prove that $M_1 \reduc_{E_M} M_2$
(detailed in \cite{tacas,fi2022}). We have that $m'_2 \defeq ({p_0}
= {0} \comma {p_6} = {1} \comma {a_2} = {1})$ is reachable in $M_2$, which means that every
solution to the system $p_0 = 0 \comma p_6 = 1 \comma p_4 = p_5 \comma p_1 + p_2
= 1 \comma p_3 + p_4 = 1$ gives a reachable marking of $M_1$. Moreover, every
solution such that $p_i \geq 1$ and $p_j \geq 1$ gives a witness that $p_i \conc
p_j$. For instance, $p_1, p_4, p_5$ and $p_6$ are certainly concurrent together.
We should exploit the fact that, under some assumptions about $E$, we can find
all such ``pairs of variables'' without the need to explicitly solve systems of
the form $E\comma \maseq{m}$; just by looking at the structure of $E$.

For this current work, we do not need to explain how to derive or check that an
equivalence statement is correct in order to describe our method. In practice,
we start from an initial net, $(N_1, m_1)$, and derive $(N_2, m_2)$ and $E$
using a combination of several structural reduction rules. You can find a
precise description of our set of rules in~\cite{berthomieu_counting_2019} and a
proof that the result of these reductions always leads to a valid
$E$-equivalence in~\cite{tacas,fi2022}. The system of linear equations obtained using
this process exhibits a graph-like structure. In the next section, we describe a
set of constraints that formalizes this observation. This is one of the
contributions of this paper, since we never defined something equivalent in our
previous works. We show with our benchmarks (Sect.~\ref{sec:experimental}) that
these constraints are general enough to give good results on a large set of
models.


\section{Token Flow Graphs}
\label{sec:tfg}

We introduce a set of structural constraints on the equations occurring in an
equivalence statement $(N_1, m_1) \reduc_E (N_2, m_2)$. The goal is to define an
algorithm that is able to easily compute information on the concurrency relation
of $N_1$, given the concurrency relation on $N_2$, by taking advantage of the
structure of the equations in $E$.

We define the \textit{Token Flow Graph} (TFG) of a system $E$ of linear
equations as a Directed Acyclic Graph (DAG) with one vertex for each variable
occurring in $E$. Arcs in the TFG are used to depict the relation induced by
equations in $E$. We consider two kinds of arcs. Arcs for \emph{redundancy
equations}, $q \red p$, to represent equations of the form $p = q$ (or $p = q +
r + \dots$), expressing that the marking of place $p$ can be reconstructed from
the marking of $q, r, \dots$ In this case, we say that place $p$ is
\emph{removed} by arc $q \red p$, because the marking of $q$ may influence the
marking of $p$, but not necessarily the other way round. Figure~\ref{fig:red}
illustrates such reduction rules on a subpart of the net $M_1$ given in
Fig.~\ref{fig:stahl}. In this case, place $p_4$ has the same $\mathrm{Pre}$ and
$\mathrm{Post}$ relation than $p_5$, thus both places are redundant. And so, by
removing place $p_5$ we obtain the TFG on the right, corresponding to the
equation $p_5 = p_4$ (modeled by a ``black dot'' arc). 

\begin{figure}[htbp]
  \begin{subfigure}{.6\linewidth}
    \centering
    \includegraphics[width=.6\linewidth]{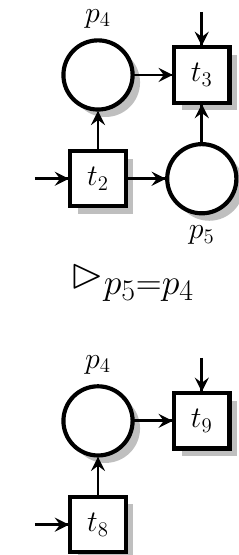}
  \end{subfigure}%
  \vrule
  \begin{subfigure}{.4\linewidth}
    \centering
    \includegraphics[width=.35\linewidth]{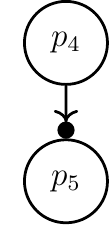}
    \vspace{20mm}
  \end{subfigure}
  \caption{Redundancy reduction applied on a subpart of the net $M_1$ from
  Fig.~\ref{fig:stahl} (left) and its corresponding TFG (right).}
  \label{fig:red}
\end{figure}

The second kind of arcs, $a \agg p$, is for \emph{agglomeration equations}. It
represents equations of the form $a = p + q$, generated when we agglomerate
several places into a new one. In this case, we expect that if we can reach a
marking with $k$ tokens in $a$, then we can certainly reach a marking with $k_1$
tokens in $p$ and $k_2$ tokens in $q$ when $k = k_1 + k_2$ (see property
Agglomeration in Lemma~\ref{lemma:forward_propagation}).  
Hence, the marking of $p$ and $q$ can be reconstructed from the marking of $a$. Thus,
we say that places p and q are removed. We also say that node $a$ is \emph{inserted}; it
does not exist in $N_1$ but may appear as a new place in $N_2$ unless it is
removed by a subsequent reduction. We can have more than two places in an
agglomeration. Figure~\ref{fig:agg} illustrates an example of such
reduction obtained by
agglomerating places $p_1, p_2$ together, in net $M_1$ of Fig.~\ref{fig:stahl}, into a
new place $a_1$. Thus, the TFG in Fig~\ref{fig:agg} (right) depicts the
obtained equation $a_1 = p_1 + p_2$ (modeled by ``white dot'' arcs).

\begin{figure}[htbp]
  \begin{subfigure}{.6\linewidth}
    \centering
    \includegraphics[width=.8\linewidth]{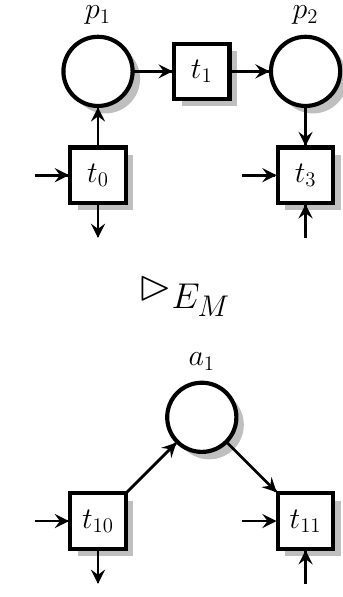}
  \end{subfigure}%
  \vrule
  \begin{subfigure}{.4\linewidth}
    \centering
    \includegraphics[width=0.7\linewidth]{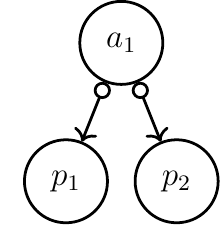}
    \vspace{20mm}
  \end{subfigure}
  \caption{Agglomeration reduction applied on a subpart of the net $M_1$ from
  Fig.~\ref{fig:stahl} (left) and its corresponding TFG (right).}
  \label{fig:agg}
\end{figure}

A TFG can also include nodes for \emph{constants}, used to express invariant
statements on the markings of the form $p + q = k$. To this end, we assume that
we have a family of disjoint sets $K(n)$ (also disjoint from place and variable
names), for each $n$ in $\Nat$, such that the ``valuation'' of a node $v \in
K(n)$ will always be $n$. We use $K$ to denote the set of all constants.
We may write $v^n$ (instead of just $v$) for a constant node whose value is $n$.

\begin{definition}[Token Flow Graph]
  \label{def:tfg}
  A TFG with set of places $P$ is a directed (bi)graph $(V, R, A)$ such that:
  \begin{itemize}
    \item $V = P \cup S$ is a set of vertices (or nodes) with $S \subset K$ a finite set of
    constants,
    \item $R \in V\times V$ is a set of \emph{redundancy arcs}, $v \red v'$,
    \item $A \in V\times V$ is a set of \emph{agglomeration arcs}, $v \agg v'$,
    disjoint from $R$.
  \end{itemize}
\end{definition}

The main source of complexity in our approach arises from the need to manage
interdependencies between $A$ and $R$ arcs, that is situations where
redundancies and agglomerations are combined. This is not something that can be
easily achieved by looking only at the equations in $E$ and thus motivates the
need for a specific data-structure.

We define several notations that will be useful in the following.  We use the
notation $v \rightarrow v'$ when we have $(v \red v')$ in $R$ or $(v \agg v')$
in $A$.  We say that a node $v$ is a \textit{root} if it is not the target of an
arc.
It is a \textit{$\circ$-leaf} when it has no output arc of the form $(v \agg v')$.
A sequence of nodes $(v_1, \dots, v_n)$ in $V^n$ is a \textit{path} if for
all $1 \leqslant i < n$ we have $v_i \rightarrow v_{i+1}$. We use the notation $v
\rightarrow^\star v'$ when there is a path from $v$ to $v'$ in the graph, or
when $v = v'$.  We write $v \agg X$ when $X$ is the largest subset $\{v_1,
\dots, v_k\}$ of $V$ such that $X \neq \emptyset$ and for all $i \in 1..k$, $v
\agg v_i \in A$. Similarly, we write $X \red v$ when $X$ is the largest,
non-empty set of nodes $\{v_1, \dots, v_k\}$ such that for all $i \in 1..k$,
$v_i \red v \in R$.
Finally, the notation $\succs{v}$ denotes the set of successors of $v$, that is:
$\succs{v} \defeq \{ v' \in V \,\mid\, v \rightarrow^\star v'\}$.

We display an example of Token Flow Graphs in Fig.~\ref{fig:HuConstruction_TFG},
where ``black dot'' arcs model edges in $R$ and ``white dot'' arcs model edges
in $A$. The idea is that each relation $X \red v$ or $v \agg X$ corresponds to
one equation $v = \sum_{v_i \in X} v_i$ in $E$, and that all the equations in
$E$ should be reflected in the TFG.
We want to avoid situations where the same place is removed more than once, or
where some place occurs in the TFG but is never mentioned in $N_1, N_2$ or $E$.
We also have the roots (without eventual constant nodes) that match to the
places in $N_2$, and the $\agg$-leaves to the places in $N_1$. All these
constraints can be expressed using a suitable notion of well-formed graph.

\begin{definition}[Well-formed TFG]
  \label{def:well-formed-tfg}
  A TFG $G = (V, R, A)$ for the equivalence statement
  $(N_1, m_1) \vartriangleright_E (N_2, m_2)$ is \emph{well-formed}
  when all the following constraints are met, where $P_1$ and $P_2$ stand
  for the set of places in $N_1$ and $N_2$:
  \begin{description}
  \item[\textbf{(T1)}] \emph{no unused names}: $V \setminus K = P_1 \cup P_2 \cup \fv{E}$;
  \item[\textbf{(T2)}] \emph{nodes in $K$ are roots}: if $v \in V \cap K$ then
    $v$ is a root of $G$;
  \item[\textbf{(T3)}] \emph{nodes can be removed only once}: it is not possible to have
    $p \agg q$ and $p' \rightarrow q$ with $p \neq p'$, or to have
    both $p \red q$ and $p \agg q$;
  \item[\textbf{(T4)}] \emph{we have all and only the equations in $E$}: we
    have $v \agg X$ or $X \red v$ if and only if the equation
    $v = \sum_{v_i \in X} v_i$ is in $E$;
  \item[\textbf{(T5)}] \textit{$G$ is acyclic};
  \item[\textbf{(T6)}] \textit{nodes in $G$ match nets}: the set of roots of
    $G$, without constants $K$, equals the set $P_2$. The set of $\circ$-leaves
    of $G$, without constants $K$, equals the set $P_1$.
  \end{description}
\end{definition}

Given a relation $(N_1, m_1) \vartriangleright_E (N_2, m_2)$, the
well-formedness conditions are enough to ensure the unicity of a TFG (up-to the
choice of constant nodes) when we set each equation to be either in $A$ or in
$R$. In this case, we denote this TFG $\tfg{E}$. In practice, we use a tool
called \reduce\ to generate the $E$-equivalence from the initial net $(N_1, m_1)$.
This tool outputs a sequence of equations suitable to build a TFG and, for each
equation, it adds a tag indicating if it is a Redundancy or an Agglomeration. We
display in Fig.~\ref{fig:HuConstruction_TFG} the equations generated by \reduce\
for the net $M_1$ given in Fig.~\ref{fig:stahl}.

\begin{figure}[htbp]
  \centering
    {
    \begin{varwidth}{\linewidth}
    \normalsize
    \begin{verbatim}
    # R |- p5 = p4
    # A |- a1 = p2 + p1
    # A |- a2 = p4 + p3
    # R |- a1 = a2
    \end{verbatim}
    \end{varwidth}}
    \vspace{\floatsep}
    \includegraphics[width=0.4\textwidth]{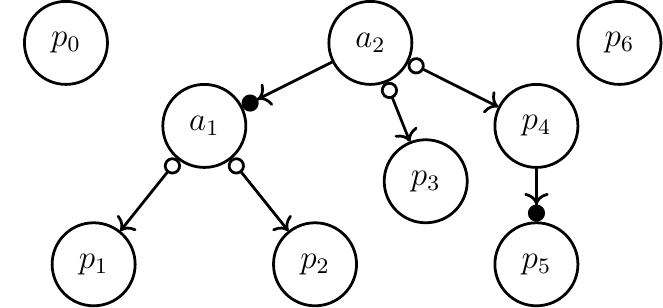}
    \caption{Equations generated from net $M_1$, in
      Fig.~\ref{fig:stahl}, and their associated TFG.}
    \label{fig:HuConstruction_TFG}
  \end{figure}
  
The constraints (T1)--(T6) are not artificial or arbitrary. In practice, we
compute $E$-equivalences using multiple steps of structural reductions, and a
TFG exactly records the constraints and information generated during these
reductions. In some sense, equations $E$ abstract a relation between the
semantics of two nets, whereas a TFG records the structure of reductions between
places.

\section{Semantics}
\label{sec:semantics}

By construction, there is a strong connection between ``systems of reduction
equations'', $E$, and their associated graph, $\tfg{E}$. We show that a similar
relation exists between solutions of $E$ and ``valuations'' of the graph (which
we call \emph{configurations} thereafter).

A \textit{configuration} $c$ of a TFG $(V, R, A)$ is a partial function from $V$
to $\Nat$. We use the notation $c(v)= \bot$ when $c$ is not defined on $v$, and
we always assume that $c(v) = n$ when $v$ is a constant node in $K(n)$.

Configuration $c$ is \textit{total} when $c(v)$ is defined for all nodes $v$ in
$V$; otherwise it is said \textit{partial}.  We use the notation $c_{\mid N}$
for the configuration obtained from $c$ by restricting its domain to the set of
places in the net $N$. We remark that when $c$ is defined over all places of $N$
then $c_{\mid N}$ can be viewed as a marking. As for markings, we say that two
configurations $c$ and $c'$ are \textit{compatible}, denoted $c \equiv c'$, if
they have same value on the nodes where they are both defined: $c(p) = c'(p)$
when $c(v) \neq \bot$ and $c'(v) \neq \bot$. (Same holds when comparing a
configuration to a marking.)  We also use $\maseq{c}$ to represent the system
$v_1 = c(v_1) \comma \dots \comma v_k = c(v_k)$ where the $(v_i)_{i \in 1..k}$
are the nodes such that $c(v_i) \neq \bot$.
We say that a configuration $c$ is \emph{well-defined} when the valuation of the
nodes agrees with the equations of $\tfg{E}$.
\begin{definition}[Well-defined configuration]
  Configuration $c$ is well-defined when for all nodes $p$ the following two
  conditions hold:
  \begin{description}
    \item[\textbf{(CBot)}] if $v \rightarrow w$ then $c(v) = \bot$ if and only if
  $c(w) = \bot$;
  \item[\textbf{(CEq)}] if $c(v) \neq \bot$ and $v \agg X$ or $X \red v$ then
  $c(v) = \sum_{v_i \in X} c(v_i)$.
  \end{description}
\end{definition}

We prove that the well-defined configurations of a TFG $\tfg{E}$ are partial
solutions of $E$, and reciprocally. Therefore, because all the variables in $E$
are nodes in the TFG (condition (T1)) we have an equivalence between solutions of
$E$ and total, well-defined configurations of $\tfg{E}$.

\begin{lemma}[Well-defined configurations are solutions]
  \label{lemma:configuration_satisfiability}
  Assume $\tfg{E}$ is a well-formed TFG for the equivalence $(N_1, m_1)
  \vartriangleright_E (N_2, m_2)$. If $c$ is a well-defined configuration of
  $\tfg{E}$ then $E\comma \maseq{c}$ is consistent. Conversely, if $c$ is a
  total configuration of $\tfg{E}$ such that $E\comma \maseq{c}$ is consistent
  then $c$ is also well-defined.
\end{lemma}

\begin{proof}
  We prove each property separately.

  Assume $c$ is a well-defined configuration of $\tfg{E}$. Since $E$ is a system
  of reduction equations, it is a sequence of equalities $\phi_1, \dots, \phi_k$
  where each equation $\phi_i$ has the form $x_i = y_1 + \dots + y_{n}$. Also,
  since $\tfg{E}$ is well-formed we have that $X_i \red v_i$ or $x_i \agg X_i$
  (only one case is possible) with $X_i = \{y_1, \dots, y_{n}\}$ for all indices
  $i \in 1..k$. We define $I$ the subset of indices in $1..k$ such that $c(x_i)$
  is defined. By condition (CBot) we have $c(x_i) \neq \bot$ if and only if
  $c(v) \neq \bot$ for all $v \in X_i$. Therefore, if $c(x_i) \neq \bot$, we
  have by condition (CEq) that $\phi_i\comma \maseq{c}$ is consistent. Moreover,
  the values of all the variables in $\phi_i$ are determined by $\maseq{c}$
  (these variables have the same value in every solution). As a consequence, the
  system combining $\maseq{c}$ and the $(\phi_i)_{i \in I}$ has a unique
  solution. On the opposite, if $c(x_i) = \bot$ then no variables in $\phi_i$
  are defined by $\maseq{c}$. Nonetheless, we know that system $E$ is
  consistent. Indeed, by property of $E$-equivalence, we know that $E\comma
  \maseq{m_1}$ has solutions, so it is also the case with $E$. Therefore, the
  system combining the equations in $(\phi_i)_{i \notin I}$ is consistent. Since
  this system shares no variables with the equations in $(\phi_i)_{i \in I}$, we
  have that $E\comma \maseq{c}$ is consistent.

  For the second case, we assume $c$ total and $E\comma \maseq{c}$ consistent.
  Since $c$ is total, condition (CBot) is true ($c(v) \neq \bot$ for all nodes
  in $\tfg{E}$). Assume we have $(N_1, m_1) \reduc_E (N_2, m_2)$. For condition
  (CEq), we rely on the fact that $\tfg{E}$ is well-formed (T4). Indeed, for all
  equations in $E$ we have a corresponding relation $X \red v$ or $v \agg X$.
  Hence $E\comma \maseq{c}$ consistent implies that $c(v) = \sum_{w \in X}
  c(w)$.
\end{proof}

We can prove several properties related to how the structure of a TFG constrains
possible values in well-defined configurations. These results can be thought of
as the equivalent of a ``token game'', which explains how tokens can propagate
along the arcs of a TFG. This is useful in our context since we can assess that
two nodes are concurrent when we can mark them in the same configuration. (A
similar result holds for finding pairs of nonconcurrent nodes.)

Our first result shows that we can always propagate tokens from a node to its
children, meaning that if a node has a token, we can find one in its
successors (possibly in a different well-defined configuration).
Property (Backward) states a dual result; if a child node is marked then one of
its parents must be marked.

\begin{lemma}[Token propagation]\label{lemma:forward_propagation} Assume
  $\tfg{E}$ is a well-formed TFG for the equivalence $(N_1, m_1)
  \vartriangleright_E (N_2, m_2)$ and $c$ a well-defined configuration of
  $\tfg{E}$.\\[-1.5em]
  \begin{description}
  \item[\textbf{(Agglomeration)}] if $p \agg \{ q_1, \dots, q_k \}$ and $c(p) \neq \bot$
    then for every sequence $(l_i)_{i \in 1..k}$ of $\Nat^k$ such that
    $c(p) = \sum_{i \in 1..k} l_i$, we can find a well-defined configuration $c'$ such that
    $c'(p) = c(p)$, and $c'(q_i) = l_i$ for all $i \in 1..k$, and $c'(v) = c(v)$
    for every node $v$ not in $\succs{p}$.
  
  \item[\textbf{(Forward)}] if $p, q$ are nodes such that $c(p) \neq \bot$ and $p
  \rightarrow^\star q$ then we can find a well-defined configuration $c'$ such
  that $c'(q) \geqslant c'(p) = c(p)$ and $c'(v) = c(v)$ for every node $v$ not
  in $\succs{p}$.
  
  \item[\textbf{(Backward)}] if $c(p) > 0$ then there is a root $v$ such that $v
    \rightarrow^\star p$ and $c(v) > 0$.
  \end{description}
\end{lemma}

\begin{proof}
  We prove each property separately.
  Without loss of generality, we assume there exists an (arbitrary) total ordering on nodes.
  \\[-1em]
  
  \noindent\textbf{Agglomeration}:
  let $p$ be a node such that $p \agg X$, with $X = \{ q_1, \dots, q_k \}$,
  and let $(l_1, \dots, l_k) \in \Nat^k$ be a sequence such that
  $c(p) = \sum_{i \in 1 ..  k} l_i$.
  We define configuration $c'$ as a recursive function.
  The base cases are defined by: $c'(p) = c(p)$, for all $i \in 1..k$, $c'(q_i) = l_i$, and $c'(v) = c(v)$ for all the nodes $v$ such that
  $v \notin \succs{p}$.
  The recursive cases concern only the nodes that are successors of nodes in $X$.
  Let $w$ be such a node. It cannot be a root (since $w$ is a successor of a node in $X$),
  hence it has at least one parent $x$. We consider two cases:
  \begin{itemize}
  \item[$\cdot$] Either $x \red w$ holds: then, let $Y$ be the set of parents of $w$ (as expected, $x \in Y$).
    Property (T3) of Definition~\ref{def:well-formed-tfg} implies that $Y \red w$ holds, and we define
    $c'(w) = \sum_{y \in Y} c'(y)$.    

  \item[$\cdot$] Or $x \agg w$ holds: then $x$ is the only parent of $w$ by property (T3).
    Let $Y$ be the set of agglomeration children of $x$: we have $x \agg Y$, and $w \in Y$.
    We define $c'(w) = c'(x)$ if $w$ is the smallest node of $Y$ (according to the total ordering on nodes),
    or $c'(w) = 0$ otherwise. This entails $c'(x) = \sum_{y \in Y} c'(y)$ (where all terms are defined as zero except one defined as $c'(x)$).
    
  \end{itemize}
  Note that the recursion always implies the parents of a given node, and is therefore well-founded since a TFG is a DAG.
  It is immediate to check that $c'$ is well-defined: by construction, (CEq) is satisfied on all nodes where $c'$ is defined.\\[-1em]

  \noindent\textbf{Forward}: take a well-defined configuration $c$
  of $\tfg{E}$ and assume we have two nodes $p,q$ such that $c(p) \neq \bot$ and
  $p \rightarrow^\star q$. The proof is by induction on the length of the path
  from $p$ to $q$.
  The initial case is when $p = q$, which is trivial.
  Otherwise, assume $p \rightarrow r \rightarrow^\star q$. It is enough to find
  a well-defined configuration $c'$ such that $c'(r) \geqslant c'(p) = c(p)$.
  Since the nodes not in $\succs{p}$ are not in the paths from $p$ to $q$, we
  can ensure $c'(v) = c(v)$ for any node $v$ not in $\succs{p}$.
  The proof proceeds by a case analysis on $p \rightarrow r$:

  \begin{itemize}
    \item[$\cdot$] Either $X \red r$ with $p \in X$. Then by (CEq) we have $c(r)
    = c(p) + \sum_{v \in X, v \neq p} c(v) \geqslant c(p)$ and we can choose
    $c'= c$.
    \item[$\cdot$] Or we have $p \agg X$ with $r \in X$. By (Agglomeration) we
    can find a well-defined configuration $c'$ such that $c'(r) = c'(p) = c(p)$
    (and also $c'(v) = 0$ for all $v \in X \setminus \{ r\}$).
  \end{itemize}

  \noindent\textbf{Backward}: let $c$ be a well-defined configuration of
  $\tfg{E}$ and assume we have $c(p) > 0$. The proof is by reverse structural
  induction on the DAG. If $p$ is a root, the result is immediate ; otherwise,
  $p$ has at least one parent $q$ such that $q \rightarrow p$. As above, we
  proceed by case analysis.

  \begin{itemize}
  \item[$\cdot$] Either $X \red p$ with $q \in X$.  By (CEq), we have $c(p) =
    \sum_{v \in X} c(v) > 0$. Hence there must be at least one node $q'$ in $X$
    such that $c(q') > 0$ and we conclude by induction hypothesis on $q'$, which
    is a parent of $p$.
  \item[$\cdot$] Or we have $q \agg X$ with $p \in X$. By (CEq), we have $c(q) =
    \sum_{v \in X} c(v) \geqslant c(p) > 0$, and we conclude by induction
    hypothesis on $q$, which is again a parent of $p$.    
  \end{itemize}

  \quad
\end{proof}

Until this point, none of our results rely on the properties of $E$-equivalence.
We now prove that there is an equivalence between the reachable markings --of
$N_1$ and $N_2$-- and configurations of $\tfg{E}$. More precisely, we prove
(Theorem~\ref{th:configuration_reachability}) that every reachable marking in
$N_1$ or $N_2$ can be extended into a well-defined configuration of $\tfg{E}$.
This entails that we can reconstruct all the reachable markings of $N_1$ by
looking at well-defined configurations obtained from the reachable markings of
$N_2$. Our algorithm for computing the concurrency relation (see
Sect.~\ref{sec:algorithm}) will be a bit smarter since we do not need to
enumerate exhaustively all the markings of $N_2$. Instead, we only need to know
which roots can be marked together.

\begin{theorem}[Configuration reachability]\label{th:configuration_reachability}
  Assume $\tfg{E}$ is a well-formed TFG for the equivalence $(N_1, m_1)
  \vartriangleright_E (N_2, m_2)$. If $m$ is a marking in $R(N_1, m_1)$ or
  $R(N_2, m_2)$ then there exists a total, well-defined configuration $c$ of
  $\tfg{E}$ such that $c \equiv m$. Conversely, given a total, well-defined
  configuration $c$ of $\tfg{E}$, marking $c_{\mid N_1}$ is reachable in
  $(N_1, m_1)$ if and only if $c_{\mid N_2}$ is reachable in $(N_2, m_2)$.
\end{theorem}

\begin{proof}
  We prove each point separately.
  
  First, we take a marking $m$ in $R(N_1, m_1)$
  (the case $m$ in $R(N_2, m_2)$ is similar). By property (A1) of
  Definition~\ref{def:e-equivalence}, $E \comma \maseq{m}$ is consistent. Hence, it
  admits a non-negative integer solution $c$, meaning a valuation for all the
  variables and places in $\fv{E}, N_1$ and $N_2$ such that $E\comma \maseq{c}$
  is consistent and $c(p) = m(p)$ if $p \in N_1$.
  We may freely extend $c$ to include constants in $K$ (whose values are fixed),
  thus according to condition (T1) of Definition~\ref{def:well-formed-tfg}, this
  solution $c$ is defined over all the nodes of $\tfg{E}$. It is well-defined by
  virtue of Lemma~\ref{lemma:configuration_satisfiability}.

  For the converse property, we assume that $c$ is a total, well-defined
  configuration of $\tfg{E}$ and that $c_{\mid N_1}$ is in $R(N_1, m_1)$ (the
  case $c_{\mid N_2}$ in $R(N_2, m_2)$ is again similar).
  Since $c$ is a well-defined configuration, from
  Lemma~\ref{lemma:configuration_satisfiability} we have that $E\comma
  \maseq{c}$ is consistent. Therefore we have that $E\comma \maseq{c_{\mid
  N_1}}\comma \maseq{c_{\mid N_2}}$ is consistent. By condition (A3) of
  Definition~\ref{def:e-equivalence}, we have $c_{\mid N_2}$ in $R(N_2, m_2)$, as
  needed.
\end{proof}

The second result of this theorem justifies the following definition.

\begin{definition}[Reachable configuration]
  \label{def:reachable-config}
  Configuration $c$ is reachable for an equivalence statement $(N_1, m_1)
  \vartriangleright_E (N_2, m_2)$ if $c$ is total, well-defined and $c_{\mid N_1} \in
  R(N_1, m_1)$  (resp. $c_{\mid N_2} \in R(N_2, m_2)$).
\end{definition}

The previous fundamental results demonstrate the possibilities of TFGs to reason
about the state space of the initial net from the one of the reduced net, and
vice versa.

\section{Marking reachability}
\label{sec:reachability}

We illustrate the benefit of Token Flow Graphs by describing a
simple model-checking algorithm. The goal is to decide if a marking, say $m_1'$, is reachable in the
initial net $(N_1, m_1)$, by checking a reachability property  on the
smaller net $(N_2,m_2)$.
We start by proving some auxiliary results.
\begin{lemma}[Unicity of marking reduction]\label{lemma:marking_reduction_unicity}
  Assume $\tfg{E}$ is a well-formed TFG for the equivalence $(N_1, m_1)
  \vartriangleright_E (N_2, m_2)$.
  Given a marking $m_1'$ of $N_1$ there exists at most one total, well-defined
  configuration $c$ such that $c \equiv m_1'$.
\end{lemma}
\begin{proof}
  Let $c_1$ and $c_2$ be two total, well-defined configurations such that
  $c_1 \equiv m_1'$ and $c_2 \equiv m_1'$. 
  Let $X$ be the set of nodes $x$ such that $c_1(x) \neq c_2(x)$.
  By contradiction, we assume $X$ is not empty (that is, we assume $c_1 \neq c_2$).
  For each node $x$ in $X$, we know that $x$ does not belong to $P_1$, since $c_1$ and $c_2$ agree on $m_1'$.
  Consequently, by virtue of property (T6) of Definition~\ref{def:well-formed-tfg}, each $x$ in $X$ admits an output $x \agg y$.
  More generally, for each $x$ in $X$, we have $x \agg Y$ for some non-empty set $Y$.
  We consider an element $x_0$ of $X$ such that $x_0 \agg Y_0$ holds with $Y_0$ disjoint from $X$ (such an element $x_0$ necessarily exists if $X$ is not empty,
  otherwise $X$ would contain a cycle of $\circ$-arcs, which is forbidden by the acyclic property (T5) of the well-formed TFG).
  Then, since $c_1$ is well-defined, we know that $c_1(x_0) = \sum_{y \in Y_0} c_1(y)$ by property (CEq).
  However, we have $c_1(y) = c_2(y)$ for all $y \in Y_0$ since $Y_0$ is disjoint from $X$.
  Hence, $c_1(x_0) = \sum_{y \in Y_0} c_2(y) = c_2(x_0)$, which contradicts $x_0 \in X$.
  As a conclusion, $X$ must be empty, that is $c_1 = c_2$. There can be at most one well-defined configuration.

\end{proof}

Then, as a corollary of this lemma and of Theorem~\ref{th:configuration_reachability}, we 
get:
\begin{theorem}[Reachability decision]
  \label{th:reachability_decision}
  Assume $\tfg{E}$ is a well-formed TFG for the equivalence $(N_1, m_1)
  \vartriangleright_E (N_2, m_2)$. Deciding if a marking $m'_1$ is reachable in
  $R(N_1, m_1)$ amounts to construct a total, well-defined configuration
  $c$ such that $c \equiv m'_1$ and then check if $c_{\mid N_2}$ is reachable in $R(N_2, m_2)$. 
\end{theorem}

Hence, given an equivalence $(N_1, m_1) \vartriangleright_E (N_2, m_2)$
and the associated TFG $\tfg{E}$, we first extend our marking of interest
$m_1'$ into a total well-defined configuration $c$, as done by Algorithm~\ref{alg:reachable}, next.
(Lemma~\ref{lemma:marking_reduction_unicity} ensures that if such configuration exists, then it is unique.)
As stated in Theorem~\ref{th:reachability_decision}, if $c$ restricted to $N_2$ is a marking reachable in $(N_2, m_2)$,
then $m_1'$ is reachable in $(N_1, m_1)$.
Otherwise, $\neg \maseq{m_1'}$ is an invariant on $R(N_1, m_1)$.

We illustrate this algorithm by taking two concrete examples on the marked net
$M_1$ given in Fig.~\ref{fig:stahl}. Assume we want to decide if marking $m_1'
\defeq ({p_0} = {0} \comma {p_1} = {2} \comma {p_2} = {0} \comma {p_3} = {1}
\comma {p_4} = {1} \comma {p_5} = {1} \comma{p_6} = {0})$ is reachable in $(N_1,
m_1)$, for $m_1$ as depicted in Fig.~\ref{fig:stahl}. This marking can be
extended into a total, well-defined configuration $c$, with $c(a_1) = c(a_2) =
2$. And so, deciding of the reachability of marking $m_1'$ in $(N_1, m_1)$ is
equivalent to decide if marking $m_2' \defeq ({p_0} = {0} \comma {a_2} = {2}
\comma{p_6} = {0})$ is reachable in $(N_2, m_2')$ (which it is not).
Observe that $m_1'$ would be reachable if the initial marking $m_1$ was
$({p_0} = 2, {p_6} = 1)$ and the other places empty.
Conversely, assume our
marking of interest is $m_1''$ such that $m_1''(p_4) = 2$ and $m_1''(p_1) =
m_1''(p_2) = 0$. It is not possible to extend this marking into a well-defined
configuration $c$, since $c(a_1) = m_1''(p_1) + m_1''(p_2) = 0$ and $c(a_1) =
c(a_2) > m_1''(p_4)$. In this case, we directly obtain that $m_1''$ is not
reachable in $(N_1, m_1)$, for every initial marking $m_1$.

The \texttt{Reachable} function (see Algorithm~\ref{alg:reachable}) is a direct
implementation of Theorem~\ref{th:reachability_decision}: it builds a total
configuration $c$, then checks that it is well-defined (we omit the function
\textit{well-defined}, which is obvious), and finally finds out if $c_{\mid
N_2}$ is reachable in $(N_2,m_2)$.
This algorithm relies on the recursive procedure \texttt{BottomUp}
(see Algorithm~\ref{alg:bottomup}) which extends the marking of interest $m_1'$ into
a total, well-defined configuration if there is one.

\begin{algorithm}[htbp]
  \caption{Reachable($m_1'$, $\tfg{E}$, $(N_2, m_2)$)}\label{alg:reachable}
  \begin{algorithmic}[1]
    \State \textbf{In:}
    \begin{tabular}[t]{l@{}c@{ }l}
      $m_1'$ & : & a marking of $N_1$ \\
      $(N_2, m_2)$ & : & reduced net such that \\
       & & $(N_1, m_1) \vartriangleright_E (N_2, m_2)$ holds \\
      $\tfg{E}$ & : & well-formed TFG for \\
      & & the $E$-equivalence above
    \end{tabular}

    \vskip 1ex
    
    \noindent
    \textbf{Out:} a boolean indicating if $m_1' \in R(N_1, m_1)$
    \vskip 2ex

    \State $c \gets \vec{\bot}$  
    \hfil \acomment{;; c is a configuration for $\tfg{E}$}

    \vskip 1ex

    \State \textbf{for all} $p \in P_1$ \textbf{do} $c[p] \gets m_1[p]$ \label{line:reach-init1}
    \State \textbf{for all} $v^n \in K$ \textbf{do} $c[v] \gets n$ \label{line:reach-init2}

    \vskip 1ex
    
    \State \acomment{;; $\tfg{E}$ is $(V,R,A)$, as in Definition~\ref{def:tfg}}
    \State \textbf{for all} $v \in V$ \textbf{do} BottomUp$(c, v, \tfg{E})$ \label{line:reach-bup}

    \vskip 1ex

    \State \textbf{return} well-defined$(c) \land c_{\mid N_2} \in R(N_2, m_2)$ \label{line:reach-return}
  \end{algorithmic}
\end{algorithm}

\begin{algorithm}[htbp]
  \caption{BottomUp($c, v, \tfg{E}$)}\label{alg:bottomup}
  \begin{algorithmic}[1]
    \State
    \begin{tabular}[t]{l@{ }l}
      \textbf{In:} & $\tfg{E}$ : the TFG structure \\
      & $v$, a node in $\tfg{E}$\\
      \textbf{In out:} & $c$, a partial configuration of $\tfg{E}$ \\
      \textbf{Post:} & $c$ is defined for all nodes of $\succs{v}$
    \end{tabular}

    \vskip 1ex
    
    \For{\textbf{all} $v'$ \textbf{such that} $v \rightarrow v'$}
    \State BottomUp$(c, v', \tfg{E})$ \label{line:bup-rec}
    \EndFor

    \vskip 1ex
    \State \textbf{if} {$v \agg X$} \textbf{then} $c[v] \gets \sum_{v' \in X} c[v']$ \label{line:bup-update}
  \end{algorithmic}
\end{algorithm}

We note that the second algorithm, which is recursive, always terminates since
it simply follows the TFG structure. We still have to prove that
Algorithm~\ref{alg:reachable} always returns the correct answer.
\begin{proof}
We consider two cases:
\begin{itemize}
\item[] \textit{Case C1}: the algorithm returns false because $c$ is not
  well-defined (line~\ref{line:reach-return}). In this case we show, next, that
  no well-defined configuration $c$ extending $m'_1$ exists, and thus $m'_1$ is
  not reachable by Theorem~\ref{th:configuration_reachability}.
  
\item[] \textit{Case C2}: the algorithm returns the value of $c_{\mid N_2} \in
R(N_2, m_2)$. In this case, thanks to Theorem~\ref{th:reachability_decision}, it
suffices to show that $c$ is total, well-defined, and extends $m'_1$.
\end{itemize}
We start by case C2, which basically states the soundness of the algorithm. Let
us show that $c$ is total: for every node~$v$, if $v$ is a constant, or if it
belongs to $P_1$, it is set by lines~\ref{line:reach-init1}
and~\ref{line:reach-init2} of Algorithm~\ref{alg:reachable}. Otherwise, $v$ is
not a $\circ$-leaf (by property (T6) of Definition~\ref{def:well-formed-tfg}), hence
it is set by line~\ref{line:bup-update} of Algorithm~\ref{alg:bottomup}, which
is invoked on every node of the TFG (line~\ref{line:reach-bup},
Algorithm~\ref{alg:reachable}).
Additionally, $c$ is well-defined, because it passed the test
line~\ref{line:reach-return}. It also extends $m'_1$ by consequence of
line~\ref{line:reach-init1} (those values are not overwritten later in
Algorithm~\ref{alg:bottomup}, because of property (T6)).

Case C1 states the completeness of the algorithm. By contraposition, we assume
there exists a total well-defined configuration $c'$ extending $m'_1$, and show
by induction on the recursive calls to \texttt{BottomUp}, that the algorithm builds a
configuration $c$ equal to $c'$. More precisely, we show that an invocation of
$BottomUp(c,v,\tfg{E})$ returns a configuration $c$ that coincides with $c'$ on
all nodes of $\succs{v}$. The initial configuration $c$ built by the algorithm
extends $m'_1$ and sets the constants (lines~\ref{line:reach-init1}
and~\ref{line:reach-init2} of Algorithm~\ref{alg:reachable}). Then, for any
invocation of $BottomUp(c,v,\tfg{E})$, for every node $w$ in $\succs{v}$, if $w$
is a node in $P_1$, then $c(v) = c'(v)$ holds immediately. Otherwise, $w$ is not
a $\circ$-leaf. If $w \neq v$, then $w$ is in $\succs{u}$ for some child $u$ of
$v$, and $c(v) = c'(v)$ holds by induction hypothesis on the recursive call to
$BottomUp(c,u,\tfg{E})$ that occurred on line~\ref{line:bup-rec}. If $w = v$,
then $c(v)$ is set by line~\ref{line:bup-update}, that is $c(v) = \sum_{v' \in
X} c(v')$. By induction hypothesis, we have $c(v') = c'(v')$ for all $v'$
children of $v$ (the recursive call occurred also on line~\ref{line:bup-rec}).
Hence, $c(v) = \sum_{v' \in X} c'(v') = c'(v)$ by property (Ceq) since $c'$ is
well-defined. Consequently, $c(w) = c'(w)$ for all $w$ in $\succs{v}$. As a
result, $c = c'$, since \texttt{BottomUp} is invoked on all nodes of the TFG.
\end{proof}

\paragraph{State space partition}

We can use the previous results to derive an interesting result about
the state space of equivalent Petri nets, in the case where the
associated TFG is well-formed. Indeed, we can prove that, in this
case, we can build a partition of the reachable markings of
$(N_1, m_1)$ that is in bijection with the reachable markings of
$(N_2, m_2)$.

Given a marking $m_2'$ of the reduced net $N_2$ we define $Inv(m_2')$ as the set
of  markings of the initial net $N_1$ which are in relation with $m'_2$. 

$$Inv(m_2') \defeq \{ c_{\mid N_1}
\mid c \text{ total, well-defined } \land c \equiv m_2'\}$$  

\begin{theorem}[State space partition]
  Assume $\tfg{E}$ is a well-formed TFG for the equivalence $(N_1, m_1)
  \vartriangleright_E (N_2, m_2)$. The family of sets $P \defeq \{ Inv(m_2')
  \mid m_2' \in R(N_2, m_2) \}$ is a partition of $R(N_1, m_1)$.
\end{theorem}

\begin{proof}
  The set $P$ is partition as a consequence of the following points:

  \noindent\textbf{No empty set in P}:
  for any $m_2'$ in $R(N_2, m_2)$,
  by Theorem~\ref{th:configuration_reachability},
  there exists a total, well-defined configuration $c$ such that $c \equiv m$.
  Thus, $Inv(m_2')$ is not empty. This implies $\emptyset \notin P$.

  \noindent\textbf{The union $\union_{A \in P} A$ covers $R(N_1,m_1)$}:
  take $m_1'$ in $R(N_1, m_1)$. From
  Theorem~\ref{th:configuration_reachability} there exists a total, well-defined
  configuration $c$ such that $c \equiv m$ and $c_{\mid N_2} \in R(N_2, m_2)$.
  Hence, the sets in $P$ cover $R(N_1, m_1)$.

  \noindent\textbf{Pairwise disjoint}: take two different markings $m_2'$ and
  $m_2''$ in $R(N_2, m_2)$. From Lemma~\ref{lemma:marking_reduction_unicity} we
  have $Inv(m_2') \cap Inv(m_2'') = \emptyset$ since every marking of $(N_1,
  m_1)$ can be extended into at most one possible configuration $c$.
\end{proof}

As a corollary, when a marked net $(N_1, m_1)$ can be partially reduced,
we know how to partition its state space into a union of disjoint
\textit{convex sets}; meaning sets of markings defined as solutions to
a system of linear equations.


\section{Concurrency relation}
\label{sec:algorithm}

In this section, we present an acceleration algorithm to compute the concurrency
relation of a Petri net using TFGs. The philosophy is quite different from the
previous algorithm for reachability decision. Here, given an equivalence
statement $(N_1, m_1) \vartriangleright_E (N_2, m_2)$ the idea is to compute the
concurrency relation directly on the reduced net $(N_2, m_2)$, and track the
information back to the initial net $(N_1, m_1)$ using the corresponding TFG
$\tfg{E}$. 

In the following, we will focus on safe nets. Fortunately, our reduction rules
preserve safeness (see Corollary~\ref{corollary:safeness_preservation}). Hence,
we do not need to check if $(N_2, m_2)$ is safe when $(N_1, m_1)$ is.
The fact that the nets are safe has consequences on configurations.  

\begin{lemma}[Safe configurations]\label{lemma:safe_configurations} Assume
  $\tfg{E}$ is a well-formed TFG for $(N_1, m_1) \reduc_E (N_2, m_2)$ with
  $(N_1, m_1)$ a safe Petri net. Then for every total,
  well-defined configuration $c$ of $\tfg{E}$ such that $c_{\mid N_1}$ reachable
  in $(N_1,m_1)$, and every node $v$ (not in $K$), we have $c(v) \in \{0, 1\}$.
\end{lemma}

\begin{proof}
  We prove the result by contradiction. Take a total, well-defined configuration
  $c$ such that $c_{\mid N_1}$ is reachable in $(N_1, m_1)$ and a node $v$ such
  that $c(v) > 1$. Since $(N_1,m_1)$ is safe, $v$ does not belong to $P_1$. From
  property (T6) of Definition~\ref{def:well-formed-tfg} we have that $v$ is not
  a $\circ$-leaf, thus $v$ has some output arcs such that $v \agg X$.
  Additionally, property (T6) also entails that there exists a place $p$ of $N_1$
  such that $v \to^\star p$. By Lemma~\ref{lemma:forward_propagation} (Forward),
  we can find a well-defined configuration $c'$ of $\tfg{E}$ such that $c'(p)
  \geqslant c'(v) = c(v) > 1$ and $c'(w) = c(w)$ for every node $w$ not in
  $\succs{v}$.
  The latter implies $c'_{\mid N_2} = c_{\mid N_2}$. Then, by
  Theorem~\ref{th:configuration_reachability}, $c'_{\mid N_1}$ is reachable in
  $(N_1,m_1)$. However,  $c'(p) > 1$ is in contradiction with the safeness of
  $(N_1,m_1)$.
\end{proof}

\begin{corollary}[Safeness preservation]
  \label{corollary:safeness_preservation}
  Assume $\tfg{E}$ is a well-formed TFG for $(N_1, m_1) \reduc_E (N_2, m_2)$.
  If $(N_1, m_1)$ is safe then $(N_2, m_2)$ is safe.
\end{corollary}

We base our approach on the fact that we can extend the notion of concurrent
places (in a marked net), to the notion of concurrent nodes in a TFG, meaning
nodes that can be marked together in a reachable configuration (as defined in Definition~\ref{def:reachable-config}).

By Theorem~\ref{th:configuration_reachability}, if we take reachable markings in
$N_2$---meaning we fix the values of roots in $\tfg{E}$---we can find places of
$N_1$ that are marked together by propagating tokens from the roots to the
leaves (Lemma~\ref{lemma:forward_propagation}). In our algorithm, next, we show
that we can compute the concurrency relation of $N_1$ by considering two
cases: (1) we start with a token in a single root $p$, with $p$ nondead, and
propagate this token forward until we find a configuration with two places in
$N_1$ marked together (which is basically due to some redundant places);
or (2) we do the same but placing a token in two separate
roots, $p_1, p_2$, such that $p_1 \conc p_2$.

\subsection{Description of the algorithm}

We assume that $\tfg{E}$ is a well-formed TFG for the relation $(N_1, m_1)
\vartriangleright_E (N_2, m_2)$. We use symbol $\conc_2$ for the concurrency
relation on $(N_2, m_2)$ and $\conc_1$ on $(N_1, m_1)$.  The set of nodes of $\tfg{E}$ is $P$.

We define an algorithm that takes as inputs a well-formed TFG $\tfg{E}$ plus the
concurrency relation $\conc_2$ on the net $(N_2, m_2)$, and outputs the
concurrency relation $\conc_1$ on $(N_1, m_1)$. Actually, our algorithm
computes a \emph{concurrency matrix}, $\Conc$, that is a symmetric matrix
such that
$\Conc[v,w] = 1$ when the nodes $v, w$ can be marked together in a reachable
configuration, and $0$ otherwise. We prove (Theorem~\ref{th:matrix}) that the
relation induced by $\Conc$ matches with $\conc_1$ on $N_1$.
Our algorithm can be pragmatically interrupted after a given time limit,
it then returns a partial relation $\conc_2$. Undefined cases
are written $\Conc[v,w] = \bullet$ in matrix $\Conc$, which is then qualified as
\textit{incomplete}.

The complexity of computing the concurrency relation is highly dependent on the
number of places in the net. For this reason, we say that our algorithm performs
some sort of ``dimensionality reduction'', because it allows us to solve a
problem in a high-dimension space (the number of places in $N_1$) by solving it
first on a lower dimension space (since $N_2$ may have far fewer places) and
then transporting back the result to the original net.
In practice, we compute the concurrency relation on $(N_2, m_2)$ using the tool
\caesar\ from the CADP toolbox; but we can rely on any kind of ``oracle'' to
compute this relation for us. This step is not necessary when the initial net is
fully reducible, in which case the concurrency relation for $N_2$ is trivial and
all the roots in $\tfg{E}$ are constants.

To simplify our notations, we assume that $v \conc_2 w$ when $v$ is a constant
node in $K(1)$ and $w$ is nondead. On the opposite, $v \mathbin{\#}_2 w$ when
$v \in K(0)$ or $w$ is dead. 

Our algorithm is divided into two main functions, \ref{alg:change_of_dimension}
and~\ref{alg:token_propagation}.
It also implicitly relies on an auxiliary function that returns the
successors $\succs{x}$ for a given node $x$ (we omit the details).
In the main function, \texttt{Matrix}, we
iterate over the nondead roots of $\tfg{E}$ and recursively propagates
the information that node $v$ is nondead: the call to \texttt{Propagate} in
line~\ref{line:alive} updates the
concurrency matrix $\Conc$ by finding all the concurrent nodes
that arise from a unique root $v$. We can prove all such cases arise from redundancy arcs with
their origin in $\succs{v}$.
More precisely, we prove in
Lemma~\ref{lemma:concurrency_redundancy_nodes} that if
$v \red w$ holds, then the nodes in the set
$\succs{v} \setminus \succs{w}$ are concurrent to all the nodes in $\succs{w}$.
This is made explicit in the \texttt{for} loop, line~\ref{line:propag-pairs}
of Algorithm~\ref{alg:token_propagation}.
Next, in the
second \texttt{for} loop of \texttt{Matrix}, we compute the concurrent nodes
that arise from two distinct nondead roots $(v, w)$. In this case, we can prove
that all the successors of $v$ are concurrent with successors of $w$: all the
pairs in $\succs{v} \times \succs{w}$ are concurrent.

\begin{algorithm}
  \caption{Matrix($\tfg{E}$, $\conc_2$)}\label{alg:change_of_dimension}
  \begin{algorithmic}[1]
    \State
    \begin{tabular}[t]{l@{ }l}
      \textbf{In:} & $\tfg{E}$ : the TFG structure \\
      & $\conc_2$: concurrency relation on $(N_2, m_2)$ \\
      \textbf{Out:} &  the concurrency matrix $\Conc$.
    \end{tabular}

    \vskip 1ex
    \State $\Conc \gets \vec{0}$ \label{line:init0} \hfil \acomment{;; C is a matrix indexed by $P \times P$}

    \vskip 1ex
    \State \acomment{;; $v$ ($\in P_2$) is nondead iff $v \conc_2 v$ holds}
    \For{\textbf{all} $v$ nondead root node in $\tfg{E}$} \label{line:for-alive}
    \State $\text{Propagate}(\tfg{E}, \Conc, v)$\label{line:alive}
    \EndFor
    \vskip 1ex

    \State \acomment{;; $v$ and $w$ ($\in P_2$) are concurrent iff $v \conc_2 w$ holds}
    \For{\textbf{all} $(v,w)$ distinct concurrent roots in $\tfg{E}$}
    \For{\textbf{all}{$(v', w') \in \succs{v} \times \succs{w}$}}
    \State $\Conc[v', w'] \gets 1$ \label{line:set-distinct-1}
    \State $\Conc[w', v'] \gets 1$ \label{line:set-distinct-2}
    \EndFor
    \EndFor
    \vskip 1ex
    \State \textbf{return} $\Conc$
  \end{algorithmic}
\end{algorithm}

\begin{algorithm}
  \caption{Propagate($\tfg{E}$, $C$, $v$)}
  \label{alg:token_propagation}
  \begin{algorithmic}[1]
    \State
    \begin{tabular}[t]{l@{ }l}
      \textbf{In:} & $\tfg{E}$ : the TFG structure \\
      & $v$: node \\
      \textbf{In out:} & the concurrency matrix $\Conc$.\\
      \textbf{Post:} & $\Conc$ contains all the concurrency \\
      & relations induced by knowing that \\
      & $v$ is nondead.
    \end{tabular}

    \vskip 1ex

    \State \acomment{;; This loop includes $\Conc[v,v] \gets 1$}    
    \For{\textbf{all} $w \in \succs{v}$}
    \State $\Conc[v,w] \gets 1$\label{line:non_dead1}
    \State $\Conc[w,v] \gets 1$\label{line:non_dead2}
    \EndFor
    \vskip 1ex
    
    \For{\textbf{all} $w$ such that $v \rightarrow w$} \label{line:for-propagate}
    \State $\text{Propagate}(\tfg{E}, \Conc, w)$ \label{line:propagate-recursion}
    \EndFor

    \vskip 1ex
    \For{\textbf{all} $w$ such that $v \red w$}
    \For{$(v', w') \in ((\succs{v} \setminus \succs{w}) \times \succs{w})$} \label{line:propag-pairs}
    \State $\Conc[v',w'] \gets 1$\label{line:redundancy_product_1}
    \State $\Conc[w',v'] \gets 1$\label{line:redundancy_product_2}
    \EndFor
    \EndFor  
  \end{algorithmic}
\end{algorithm}

We can perform a cursory analysis of the complexity of our algorithm.
We update the matrix by recursively invoking \texttt{Propagate}, along the edges of $\tfg{E}$, starting
from the roots.
(Of course, an immediate optimization consists in marking the visited nodes, so that
the function \texttt{Propagate} is never invoked twice on the same node. We do not
provide the details of this optimization, since it has no impact on soundness, completeness, or theoretical complexity.)
More precisely, we call \texttt{Propagate} only on the nodes that are nondead in
$\tfg{E}$.
Hence, our algorithm performs a number of function calls that is linear in the number of nondead
nodes.
During each call to \texttt{Propagate}, we may update at most $O(N^2)$ values in
$\Conc$, where $N$ is the number of nodes in $\tfg{E}$ (see the \texttt{for}
loop on line~\ref{line:propag-pairs}). As a result, the complexity of our
algorithm is in $O(N^3)$, given the concurrency relation $\conc_2$. This has to
be compared with the complexity of building then checking the state space of the
net, which is PSPACE. Thus, computing the concurrency relation $\conc_2$ of
$(N_2, m_2)$, with a lower dimension, and tracing it back to $(N_1, m_1)$ is
quite benefiting.

In practice, our algorithm is efficient and its execution time is often
negligible when compared to the other tasks involved when computing the
concurrency relation. We give some results on our performances in
Sect.~\ref{sec:experimental}.

\subsection{Proof of correctness}

The soundness and completeness proofs of the algorithm rely on the following definition:

\begin{definition}[Concurrent nodes]\label{def:bconc} The concurrency relation
  of $\tfg{E}$, denoted $\bconc$, is the relation between pairs of nodes in
  $\tfg{E}$ such that $v \bconc w$ holds if and only if there is a total, well-defined
  configuration $c$ where: (1) $c$ is reachable, meaning $c_{\mid N_2} \in
  R(N_2, m_2)$; and (2) $c(v) > 0$ and $c(w) > 0$.
\end{definition}

The concurrency relation $\bconc$ of $\tfg{E}$ is a generalization of both
the concurrency relation $\conc_1$ of $N_1$ and $\conc_2$ of $N_2$:
for any pair of places $(p,q) \in P_1^2$,
by Theorem~\ref{th:configuration_reachability}, we have $p \conc_1 q$ if and only if
$p \bconc q$. Similarly for $(p,q) \in P_2^2$: $p \conc_2 q$ if and only if $p \bconc q$.
We say in the latter case that $p, q$ are \emph{concurrent roots}.
As a result, $\bconc$ is symmetric
and $v \bconc v$ means that $v$ is nondead (that is, there is a valuation $c$
with $c(v) > 0$).
We can extend this notion to constants: we say that two roots $v_1, v_2$ are
concurrent when $v_1 \bconc v_2$ holds, and that root $v_1$ is nondead when
we have $v_1 \bconc v_1$. This includes cases where $v_1$ or $v_2$ are in $K(1)$ (they are
constants with value~$1$).

We prove some properties about the relation~$\bconc$ that are direct corollaries
of our token propagation properties.
For all the following results, we implicitly assume that $\tfg{E}$ is a
well-formed TFG for the relation $(N_1, m_1) \reduc_E (N_2, m_2)$, that both
marked nets are safe, and that $\bconc$ is the concurrency relation of $\tfg{E}$.

\subsubsection{Checking nondead nodes}

We start with a property (Lemma~\ref{lemma:non_dead_propagation}) stating that
the successors of a nondead node are also nondead.
Lemma~\ref{lemma:non_dead_places_completeness} provides a dual result, useful to
prove the completeness of our approach; it states that it is enough to explore
the nondead roots to find all the nondead nodes.

\begin{lemma}\label{lemma:non_dead_propagation} If $v
  \bconc v$ and $v \to^\star w$ then $w \bconc w$ and $v \bconc w$.
\end{lemma}

\begin{proof}
  Assume $v \bconc v$. This means that there is a total, well-defined
  configuration $c$ such that $c(v) > 0$ and $c_{\mid N_2} \in R(N_2, m_2)$. Now
  take a successor node of $v$, say $v \to^\star w$. By
  Lemma~\ref{lemma:forward_propagation}, we can find another reachable
  configuration $c'$ such that $c'(w) \geqslant c'(v) = c(v)$ and $c'(x) = c(x)$
  for all nodes $x$ not in $\succs{v}$. Therefore we have both $w \bconc w$ and
  $v \bconc w$.
\end{proof}

\begin{lemma}\label{lemma:non_dead_places_completeness} If $v \bconc v$ then there is
  a root $v_0$ such that $v_0 \bconc v_0$ and $v_0 \to^\star {v}$.
\end{lemma}

\begin{proof}
  Assume $v \bconc v$. Then there is a total, well-defined configuration $c$
  such that $c_{\mid N_2} \in R(N_2, m_2)$ and $c(v) > 0$. By the backward
  propagation property of Lemma~\ref{lemma:forward_propagation} we know that
  there is a root, say $v_0$, such that $c(v_0) \geq c(v)$ and $v_0 \to^\star
  v$. Hence $v_0$ is nondead in $\tfg{E}$.
\end{proof}

\subsubsection{Checking concurrent nodes}
\label{sec:check-conc-plac}

We can prove similar results for concurrent nodes instead of nondead ones. We
consider the two cases considered by function \texttt{Matrix}: when
concurrent nodes are obtained from two concurrent roots
(Lemma~\ref{lemma:concurrent_nodes_propagation}); or when they are obtained from
a single nondead root (Lemma~\ref{lemma:concurrency_redundancy_nodes}), because of
redundancy arcs. Finally, Lemma~\ref{lemma:concurrent_places_propagation}
provides the associated completeness result.

\begin{lemma}
  \label{lemma:concurrent_nodes_propagation} Assume $v,w$ are two nodes in
  $\tfg{E}$ such that $v \notin \succs{w}$ and $w \notin \succs{v}$. If $v
  \bconc w$ then $v' \bconc w'$ for all pairs of nodes $(v', w') \in \succs{v}
  \times \succs{w}$.
\end{lemma}
  
\begin{proof}
  Assume $v \bconc w$, $v \notin \succs{w}$ and $w \notin \succs{v}$.
  By definition, there exists a total, well-defined configuration $c$ such that $c(v),
  c(w) > 0$ and $c_{\mid N_2} \in R(N_2, m_2)$.

  Take a successor $v'$ in $\succs{v}$, by applying the token propagation from
  Lemma~\ref{lemma:forward_propagation} we can construct a total, well-defined
  configuration $c'$ of $\tfg{E}$ such that $c'(v') \geqslant c'(v) = c(v)$ and
  $c'(x) = c(x)$ for any node $x$ not in $\succs{v}$. Hence, $c'(w) = c(w) > 0$.

  We can use the token propagation property again, on $c'$. This gives a total,
  well-defined configuration $c''$ such that $c''(w') \geqslant c''(w) = c'(w) =
  c(w)$ and $c''(x) = c'(x)$ for any node $x$ not in $\succs{w}$.

  We still have to prove $v' \notin \succs{w}$ (which would be immediate in a tree structure, but
  requires extra proof in our DAG structure).
  Then, we will be able to conclude by observing that it implies $c''(v') = c'(v') \geqslant c(v)$
  and therefore $v'\bconc w'$ as needed.

  We prove $v' \notin \succs{w}$ by contradiction.
  Indeed, assume $v' \in \succs{w}$. Hence, $\succs{v} \cap
  \succs{w} \neq \emptyset$.  Moreover, since $E$ is a well-formed TFG, there
  must exist (condition (T3)) three nodes $p,q,r$ such that $X \red r$, $p \in
  \succs{v} \cap X$ and $q \in \succs{w} \cap X$. Like in the
  proof of Lemma~\ref{lemma:forward_propagation} we can propagate the tokens
  contained in $v,w$ to $p,q$, and obtain $c''(r) > 1$ from (CEq), which
  contradicts our assumption that the nets are safe.

\end{proof}

\begin{lemma}
  \label{lemma:concurrency_redundancy_nodes} If $v \bconc v$ and $v \red w$ then
  $v' \bconc w'$ for every pair of nodes $(v', w')$ such that $v'\in (\succs(v)
  \setminus \succs{w})$ and $w' \in \succs{w}$.
\end{lemma}

\begin{proof}
  Assume $v \bconc v$ and $v \red w$. By definition, there is a total, well-defined
  configuration $c$ such that $c(v) > 0$ and $c_{\mid N_2} \in R(N_2, m_2)$.
  Furthermore, by $v \red w$ and condition (CEq), we have $c(w) > 0$.

  Take $w'$ in $\succs{w}$. From Lemma~\ref{lemma:forward_propagation} we can
  find a total, well-defined configuration $c'$ such that $c'(w') \geqslant
  c'(w) = c(w) > 0$ and $c'(x) = c(x)$ for any node $x$ not in $\succs{w}$.
  Since $v$ is not in $\succs{w}$ we have $c'(v) = c(v)$. Likewise, places from
  $N_2$ are roots and therefore cannot be in $\succs{w}$. So we have $c'_{\mid
  N_2} \equiv c_{\mid N_2}$, which means $c'_{\mid N_2}$ is reachable in $(N_2,
  m_2)$. At this point we have $v \bconc w'$.

  Now, consider $v' \in \succs{v} \setminus \succs{w}$ with $v' \neq v$ (the expected result already holds if $v' = v$).
  Necessarily, there exists $v_0 \neq w$ such that $v \rightarrow v_0$ and $v_0 \rightarrow^\star v'$.
  We can use the forward
  propagation in Lemma~\ref{lemma:forward_propagation} on $c'$ to find a total, well-defined
  configuration $c''$ such that $c''(v_0) \geqslant c''(v) = c'(v)$ and $c''(x) =
  c(x)$ for all nodes $x$ not in $\succs{v}$, and so, $c''_{\mid N_2}$ is
  reachable in $(N_2, m_2)$. Since configuration $c''$ is well-defined we have
  (condition (CEq)) that $c''(w) \geqslant c''(v)$. 
  We consider three cases:
  \begin{itemize}
  \item[$\cdot$] Either $v_0 \notin \succs{w}$ and $w \notin \succs{v_0}$,
    and we conclude by Lemma~\ref{lemma:concurrent_nodes_propagation} that $v' \bconc w'$ holds for every node $v' \in \succs{v_0}$.
  \item[$\cdot$] Or $v_0 \in \succs{w}$: this case cannot happen since by hypothesis $v' \notin \succs{w}$ and $v_0 \rightarrow^\star v'$.

  \item[$\cdot$] Or $w \in \succs{v_0}$: by applying the same proof than the one at the end of Lemma~\ref{lemma:concurrent_nodes_propagation},
    we can show that this case leads to a non-safe marking, which is therefore excluded.
    
  \end{itemize}
  As a result, we have $v' \bconc w'$ for all $v' \in \succs{v} \setminus \succs{w}$ and all $w' \in \succs{w}$.
\end{proof}

\begin{lemma}
  \label{lemma:concurrent_places_propagation} If $v \bconc w$ holds with $v \notin \succs{w}$
  and $w \notin \succs{v}$,
  then one of the following two conditions is true.
  \begin{description}
  \item[{\textbf{(Redundancy)}}] There is a nondead node $v_0$ such that
    $v_0 \red w_0$ and either $(v, w)$ or $(w, v)$ are in
    $(\succs{v_0} \setminus \succs{w_0}) \times \succs{w_0}$.
  \item[{\textbf{(Distinct)}}] There is a pair of distinct roots $(v_0, w_0)$ such
    that $v_0 \bconc w_0$ with $v \in \succs{v_0}$ and $w \in \succs{w_0}$.
  \end{description}
\end{lemma}

\begin{proof}
  Assume $v \bconc w$. Then there is a total, well-defined
  configuration $c$ such that $c_{\mid N_2} \in R(N_2, m_2)$
  and $c(v), c(w) = 1$ (the nets are safe).
  By the backward-propagation property in Lemma~\ref{lemma:forward_propagation}
  there exists two roots $v_0$ and $w_0$ such that $c(v_0) = c(w_0) = 1$ with $v
  \in \succs{v_0}$ and $w \in \succs{w_0}$. We need to consider two cases:
  \begin{itemize}
    \item[$\cdot$] Either $v_0 \neq w_0$, that is condition (Distinct).
    \item[$\cdot$] Or we have $v_0 = w_0$. We prove that there must be a node
    $v_1$ such that $v_0 \to^\star v_1$ and $v_1 \red w_1$ with either $(v, w)$
    or $(w, v)$ in $(\succs{v_1} \setminus \succs{w_1}) \times \succs{w_1}$. We
    prove this result by contradiction.
    Indeed, if no such node exists then both $v$ and $w$ can be reached from
    $v_0$ by following only edges in $A$ (agglomeration arcs). 
    Consider $v_0 \agg Y$, there are two nodes $v',w' \in Y$ such that $v \in
    \succs{v'}$ and $w \in \succs{w'}$. Since $c$ well-defined, from (CEq)
    either $c(v') = 0$ or $c(w') = 0$. Take $c(v') = 0$ and the agglomeration
    path from $v'$ to $v$, as $v' \agg a_0 \agg \dots \agg a_n = v$ with $n \in
    \mathbb{N}$. By induction on this path, we necessarily have $c(a_i) = 0$ for
    all $i \in 0..n$, since $c$ is well-defined and a node can only have on
    parent (condition (T3)). Hence, $c(v') = 0$ that contradicts $v \bconc w$.
  \end{itemize}
\end{proof}

\subsubsection{Algorithm is sound and complete}

\begin{theorem}
  \label{th:matrix}
  If $\Conc$ is the matrix returned by a call to $Matrix(\tfg{E},
  \|)$, with $\|$ the concurrency relation between roots of $\tfg{E}$ (meaning
  $N_2$ and constants), then for all nodes $v, w$ we have $v
  \bconc w$ if and only if $\Conc[v,w] = 1$.
\end{theorem}

\begin{proof}
  First, let us remark that the call to $Matrix(\tfg{E}, \|)$
  always terminates, since the only recursion (\texttt{Propagate}) follows the DAG structure.
  We divide the proof into two different cases: first we
  show that the computation of nondead nodes (the diagonal of $\Conc$ and the nondead
  nodes of $\bconc$) is sound and complete. Next, we prove soundness and completeness
  for pairs of distinct nodes.\\[-1em]

  \noindent\textbf{Nondead places, diagonal of $\Conc$}:

  (Completeness) If $v \bconc v$ holds for some node $v$, then, by Lemma~\ref{lemma:non_dead_places_completeness},
  there exists a nondead root $v_0$ with $v \in \succs{v_0}$.
  Hence, Algorithm~\ref{alg:change_of_dimension} invokes \texttt{Propagate} on line~\ref{line:alive} with $v_0$.
  Then, all nodes in $\succs{v_0}$ are recursively visited by line~\ref{line:propagate-recursion} in Algorithm~\ref{alg:token_propagation}
  including~$v$. As a consequence, $\Conc[v,v]$ is set to $1$ on line \ref{line:non_dead1} (and remains equal to $1$ until the end of algorithm,
  since no line of the algorithm sets values of $\Conc$ to $0$ after the initialization line~\ref{line:init0}).
  This concludes the completeness part for nondead places.

  (Soundness) Conversely, assume $\Conc[v,v] = 1$ for some node $v$. We consider three subcases:
  \begin{itemize}
  \item[$\cdot$] $\Conc[v,v]$ was set on line~\ref{line:set-distinct-1} or~\ref{line:set-distinct-2} of Algorithm~\ref{alg:change_of_dimension}:
    this means that there exist two distinct concurrent roots $v_0$ and $w_0$ such that $v \in \succs{v_0} \cap \succs{w_0}$.
    Hence $v_0$ is nondead (as well as $w_0$). This implies that $v$ is nondead by Lemma~\ref{lemma:non_dead_propagation}.
    
  \item[$\cdot$] $\Conc[v,v]$ was set on line~\ref{line:non_dead1} of Algorithm~\ref{alg:token_propagation}:
    the \texttt{for} loops on line~\ref{line:for-alive} of Algorithm~\ref{alg:change_of_dimension} and on line~\ref{line:for-propagate} of Algorithm~\ref{alg:token_propagation}
    ensure that \texttt{Propagate} is only invoked on successors of nondead roots of $\tfg{E}$.
    Hence, $v$ belongs to $\succs{v_0}$ for some nondead root $v_0$, and thus $v \bconc v$ holds by Lemma~\ref{lemma:non_dead_propagation}.
    
  \item[$\cdot$] $\Conc[v,v]$ was set on line~\ref{line:redundancy_product_1} or~\ref{line:redundancy_product_2} of Algorithm~\ref{alg:token_propagation}:
    this subcase is not possible, since these lines only consider pairs $(v',w')$ of distinct nodes.
  \end{itemize}

 To conclude this first case, the algorithm is sound and complete with respect to nondead places and the diagonal of $\Conc$.
  \\[-1em]
  
 \noindent\textbf{Concurrent places}:

 (Completeness) We assume $v \bconc w$ holds for two distinct nodes $v$ and $w$.
 This implies that both $v$ and $w$ are nondead, that is $v \bconc v$
 and $w \bconc w$.
 If we have $v \in \succs{w}$, then $\Conc[v,w]$ is set to $1$ on line~\ref{line:non_dead1} or~\ref{line:non_dead2} of
 Algorithm~\ref{alg:token_propagation},
 and similarly if $w \in \succs{v}$, which is the expected result.
 Hence, we now assume that $v \notin \succs{w}$ and $w \notin \succs{w}$, and thus Lemma~\ref{lemma:concurrent_places_propagation}
 applies. We consider the two cases of the lemma:
 \begin{itemize}
 \item[$\cdot$](Redundancy): then, $\Conc[v,w]$ is set to $1$ on line~\ref{line:redundancy_product_1} or~\ref{line:redundancy_product_2} of
   Algorithm~\ref{alg:token_propagation}.
      
 \item[$\cdot$](Distinct): then $\Conc[v,w]$ is set to $1$ on line~\ref{line:set-distinct-1} or~\ref{line:set-distinct-2}
   of Algorithm~\ref{alg:change_of_dimension}.  
 \end{itemize}
 This concludes the completeness of the algorithm for concurrent places.
 
 (Soundness) We assume $\Conc[v,w] = 1$ for some distinct nodes $v$, $w$.
  We consider three subcases:
  \begin{itemize}
  \item[$\cdot$] $\Conc[v,w]$ was set on line~\ref{line:set-distinct-1} or~\ref{line:set-distinct-2} of Algorithm~\ref{alg:change_of_dimension}:
    we conclude by Lemma~\ref{lemma:concurrent_nodes_propagation} that $v \bconc w$ holds.
    
  \item[$\cdot$] $\Conc[v,w]$ was set on line~\ref{line:non_dead1} or~\ref{line:non_dead2} of Algorithm~\ref{alg:token_propagation}:
    we conclude by Lemma~\ref{lemma:non_dead_propagation}.
    
  \item[$\cdot$] $\Conc[v,w]$ was set on line~\ref{line:redundancy_product_1} or~\ref{line:redundancy_product_2} of Algorithm~\ref{alg:token_propagation}:
    we conclude by Lemma~\ref{lemma:concurrency_redundancy_nodes}.
  \end{itemize}
  This concludes the soundness of the algorithm for concurrent places.

  As a result, the algorithm is sound and complete for nondead places and for concurrent places.
\end{proof}


\subsection{Extensions to incomplete concurrency relations}

With our approach, we only ever writes 1s into the concurrency matrix $\Conc$.
This is enough since we know relation $\conc_2$ exactly and, in this case,
relation $\conc_1$ must also be complete (we can have only $0$s or $1$s in
$\Conc$). This is made clear by the fact that $\Conc$ is initialized with ${0}$s
everywhere.
We can extend our algorithm to support the case where we only have a partial
knowledge of $\conc_2$. This is achieved by initializing $\Conc$ with the
special value $\bullet$ (undefined) and adding rules that let us ``propagate
$0$s'' on the TFG, in the same way that our total algorithm only propagates
$1$s. For example, we know that if $\Conc[v,w]=0$ ($v,w$ are nonconcurrent) and
$v \agg w'$ (we know that always $c(v) \geq c(w')$ on reachable configurations)
then certainly $\Conc[w',w] =0$. Likewise, we can prove that following rule for
propagating ``dead nodes'' is sound: if $X \red v$ and $\Conc[w,w]=0$ (node $w$
is dead) for all $w \in X$ then $\Conc[v,v]=0$.

Partial knowledge on the concurrency relation can be useful. Indeed, many use
cases can deal with partial knowledge or only rely on the nonconcurrency relation
(a $0$ on the concurrency matrix). This is the case, for instance, when
computing NUPN partitions, where it is always safe to replace a $\bullet$ with a
$1$. It also means that knowing that two places are nonconcurrent is often more
valuable than knowing that they are concurrent; $0$s are better than $1$s.

We have implemented an extension of our algorithm for the case of incomplete
matrices using this idea, and we report some results obtained with it.
Unfortunately, we do not have enough space to describe the full algorithm here.
It is slightly more involved than for the complete case and is based on a
collection of six additional axioms. While we can show that the algorithm is
sound, completeness takes a different meaning: we show that when nodes $p$ and
$q$ are successors of roots $v_1$ and $v_2$ such that $\Conc[v_i,v_i] \neq
\bullet$ for all $i \in 1..2$ then necessarily $\Conc[p,q] \neq \bullet$.

In the following, we use the notation $v \bind w$ to say $\neg (v \Conc w)$;
meaning $v, w$ are nonconcurrent according to $\Conc$. With our notations, $v
\bind v$ means that $v$ is dead: there is no well-defined, reachable
configuration $c$ with $c(v) > 0$.

\subsubsection{Propagation of dead nodes}
We prove that a dead node, $v$, is necessarily nonconcurrent to all the other
nodes. Also, if all the ``direct successors'' of a node are dead then also is
the node.

\begin{lemma}
  \label{lemma:independent_nodes_propagation_1} Assume $v$ a node in $\tfg{E}$.
  If $v \bind v$ then for all nodes $w$ in $\tfg{E}$ we have $v \bind w$.
\end{lemma}

\begin{proof}
  Assume $v \bind v$. Then for any total, well-defined configuration $c$ such
  that $c_{\mid N_2}$ is reachable in $(N_2, m_2)$ we have $c(v) = 0$. By
  definition of the concurrency relation $\bconc$, $v$ cannot be concurrent to
  any node.
\end{proof}

\begin{lemma}
  \label{lemma:dead_node} Assume $v$ a node in $\tfg{E}$ such that $v
  \agg X$ or $X \red v$. Then $v \bind v$ if and only if $w \bind w$ for all
  nodes $w$ in $X$.
\end{lemma}

\begin{proof}
  We prove by contradiction both directions.
  
  Assume $v \bind v$ and take $w \in X$ such that $w \bconc w$. Then there is a
  total, well-defined configuration $c$ such that $c(w) > 0$. Necessarily, since
  $v \bind v$ we have $c(v) = 0$, which contradicts (CEq).

  Next, assume $v \bconc v$ and $w \bind w$ for every node $w \in X$. Then there
  is a total, well-defined configuration $c$ such that $c(v) > 0$. Necessarily,
  for all nodes $w \in X$ we have $c(w) = 0$, which also contradicts (CEq).
\end{proof}

These properties imply the soundness of the following three axioms:
\begin{enumerate}
\item If $\Conc[v,v] = 0$ then $\Conc[v, w] = 0$ for all node $w$ in $\tfg{E}$.
\item If $v \agg X$ or $X \red v$ and $\Conc[w,w] = 0$ for all nodes $w \in X$
  then $\Conc[v,v] = 0$.
\item If $v \agg X$ or $X \red v$ and $\Conc[v,v] = 0$ then $\Conc[w,w] = 0$ for
  all nodes $w \in X$.
\end{enumerate}

\subsubsection{Nonconcurrency between siblings}
We prove that direct successors of a node are nonconcurrent from each other (in
the case of safe nets). This is basically a consequence of the fact that $c(v) =
c(w) + c(w') + \dots$ and $c(v) \leqslant 1$ implies that at most one of $c(w)$
and $c(w')$ can be equal to $1$ when the configuration is fixed.

\begin{lemma}
  \label{lemma:independent_nodes_propagation_2} Assume $v$ a node in $\tfg{E}$
  such that $v \agg X$ or $X \red v$. For every pair of nodes $w,w'$ in $X$, we
  have that $w \neq w'$ implies $w \bar{C} w'$.
\end{lemma}

\begin{proof}
  The proof is by contradiction.  Take a pair of distinct nodes $w,w'$ in $X$
  and assume $w \bconc w'$. Then there exists a total, well-defined
  configuration $c$ such that $c(w) = 1$ and $ c(w') = 1$, with
  $c_{\mid N_2}$ reachable in $(N_2, m_2)$. Since $c$ must satisfy (CEq) we have
  $c(v) \geqslant 2$, which contradicts the fact that our nets are safe, see
  Lemma~\ref{lemma:safe_configurations}.
\end{proof}

This property implies the soundness of the following axiom:
\begin{enumerate}
  \setcounter{enumi}{3}
\item If $v \agg X$ or $X \red v$ then $\Conc[w,w'] = 0$ for all pairs of nodes
  $w, w' \in X$ such that $w \neq w'$.
\end{enumerate}

\subsubsection{Heredity and nonconcurrency}
We prove that if $v$ and $v'$ are nonconcurrent, then $v'$ must be nonconcurrent
from all the direct successors of $v$ (and reciprocally). This is basically a
consequence of the fact that $c(v) = c(w) + \dots$ and $c(v) + c(v') \leqslant
1$ implies that $c(w) + c(v') \leqslant 1$.

\begin{lemma}
  \label{lemma:independent_nodes_propagation_3} Assume $v$ a node in $\tfg{E}$
  such that $v \agg X$ or $X \red v$. Then for every node $v'$ such that $v
  \bind v'$ we also have $w \bind v'$ for every node $w$ in $X$. Conversely, if
  $w \bind v'$ for every node $w$ in $X$ then $v \bind v'$.
\end{lemma}

\begin{proof}
  We prove by contradiction each property separately.

  Assume $v \bind v'$ and take $w \in X$ such that $w \bconc v'$. Then there is
  a total, well-defined configuration $c$ such that $c(w),c(v') > 0$.
  Necessarily, since $v \bind v'$ we must have $c(v) = 0$ or $c(v') = 0$. We
  already know that $c(v') > 0$, so $c(v) = 0$, which contradicts (CEq) since $w
  \in X$.
  
  Next, assume $w \bind v'$ for all nodes $w \in X$ and we have $v \bconc v'$.
  Then there is a total, well-defined configuration $c$ such that $c(v),c(v') >
  0$. Necessarily, for all nodes $w \in X$ we have $c(w) = 0$ or $c(v') = 0$. We
  already know that $c(v') > 0$, so $c(w) = 0$ for all nodes $w \in X$, which
  also contradicts (CEq).
\end{proof}

These properties imply the soundness of the following two axioms:
\begin{enumerate}
\setcounter{enumi}{4}
\item If $v \agg X$ or $X \red v$ and $\Conc[w,v'] = 0$ for all nodes $w \in X$
  then $\Conc[v,v'] = 0$.
\item If $v \agg X$ or $X \red v$ and $\Conc[v,v'] = 0$ then $\Conc[w,v'] = 0$
  for all nodes $w$ in $X$.
\end{enumerate}


\section{Experimental results}
\label{sec:experimental}

We have implemented a new tool, called \kong, for {Koncurrent places
  Grinder}, that is in charge of performing the ``inverse transforms''
that we described in Sect.~\ref{sec:reachability}
and~\ref{sec:algorithm}. This tool is open-source, under the GPLv3
license, and is freely available on
GitHub~\footnote{\url{https://github.com/nicolasAmat/Kong}}.

We use the extensive database of models provided by the Model Checking
Contest (MCC)~\cite{mcc2019,HillahK17} to experiment with our
approach. \kong\ takes as inputs Petri nets defined using either the
Petri Net Markup Language (PNML)~\cite{hillah2010pnml}, or the
Nest-Unit Petri Net (NUPN) format~\cite{garavel_nested-unit_2019}.

We do not compute net reductions with \kong\ directly, but rather rely
on another tool, called \reduce, that is developed inside the Tina
toolbox~\cite{tinaToolbox}. We used version 2.0 of \kong\ and 3.7 of
the Tina toolbox.

\subsection{Toolchains description}

We describe the toolchains used for the marking reachability decision
procedure and for the computation of concurrency matrices.


Figure~\ref{fig:reachability_decision} depicts the toolchain used for
checking if a given marking, $m_1'$, is reachable in an input net
$(N_1, m_1)$.  In this case, marking $m_1'$ is defined in an input
file, using a simple textual format.  The tool \kong\ retrieves the
reduction system, $E$, computed with \reduce\ and uses it to project
$m_1'$ into a marking $m_2'$, if possible. If the projection returns
an error,
we know that $m_1'$ cannot be reachable. Otherwise, we call an
auxiliary tool, in this case \sift, to explore the
state space of $(N_2, m_2)$ and try to find marking $m'_2$.

\begin{figure}[htbp]
  \includegraphics[width=\linewidth]{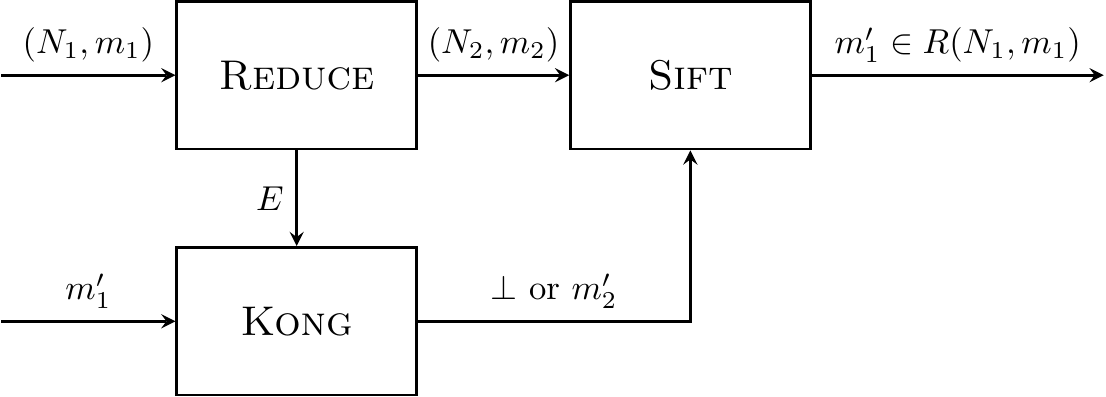}
  \caption{Toolchain of the reachability decision procedure.}
  \label{fig:reachability_decision}
\end{figure}

We describe our second toolchain in
Fig.~\ref{fig:concurrent_places}. After computing a polyhedral
reduction with \reduce, we compute the concurrency matrix of the
reduced net $(N_2, m_2)$ using \caesar, which is part
of the CADP toolbox~\cite{pbhg2021,cadp}. Our experimental results
have been computed with version v3.6 of \caesar, part
of CADP version 2022-b "Kista", published in February 2022. The tool
\kong\ takes this concurrency relation, denoted $\conc_2$, and the
reduction system, $E$, then reconstructs the concurrency relation on
the initial net.

\begin{figure}[htbp]
  \includegraphics[width=0.8\linewidth]{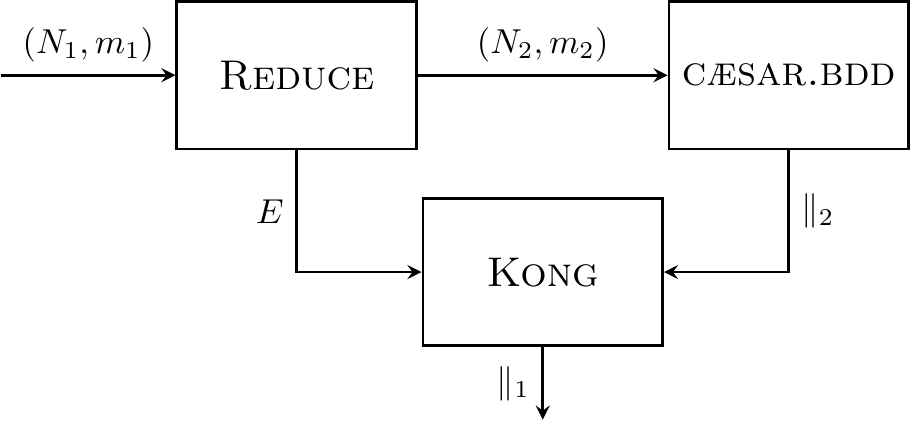}
  \caption{Toolchain of the concurrency acceleration algorithm.}
  \label{fig:concurrent_places}
\end{figure}

\subsection{Benchmarks and distribution of reduction ratios}

Our benchmark is built from a collection of $102$ models used in the MCC 2020
competition. Most models are parametrized, and therefore there can be several
different \textit{instances} for the same model. There are about $1\,000$
instances of Petri nets ($588$ that are safe), whose size vary widely, from $9$
to $50\,000$ places, and from $7$ to $200\,000$ transitions. Overall, the
collection provides a large number of examples with various structural and
behavioral characteristics, covering a lager variety of use cases.

\subsubsection{Distribution of reduction ratios}

Since we rely on how much reduction we can find in nets, we computed the
{reduction ratio} ($r$), obtained using \reduce, on all the instances (see
Fig.~\ref{fig:all_reduction}). The ratio is calculated as the quotient between
the number of places that can be removed, and the number of places in the
initial net. A ratio of $100\%$ ($r = 1$) means that the net is \emph{fully
reduced}; the residual net has no places and all the roots in its TFG are
constants.

We see that there is a surprisingly high number of models whose size
is more than halved with our approach (about $25$\% of the instances have a
ratio $r \ge 0.5$), with approximately half of the instances that can be reduced
by a ratio of $30\%$ or more. We consider two values for the reduction ratio:
one for reductions leading to a well-formed TFG (in light orange), the other for
the best possible reduction with \reduce\ (in dark blue), used for instance in the
SMPT model-checker \cite{tacas,fi2022}. The same trends can be observed for the safe nets
(Fig.~\ref{fig:safe_reduction}). 

\begin{figure}[h]
	\centering
  \begin{subfigure}[b]{\linewidth}
    \includegraphics[width=\linewidth]{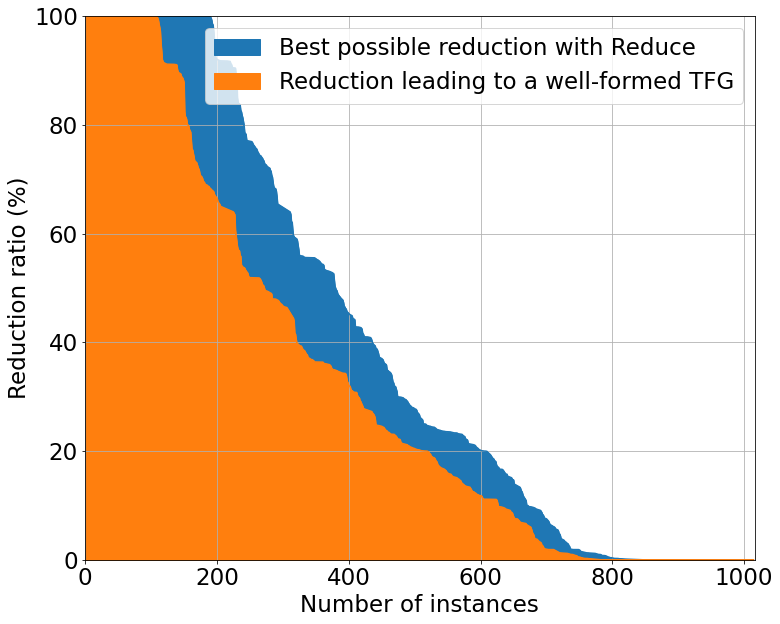}
    \caption{All instances}
    \label{fig:all_reduction}
  \end{subfigure}
  \begin{subfigure}[b]{\linewidth}
    \includegraphics[width=\linewidth]{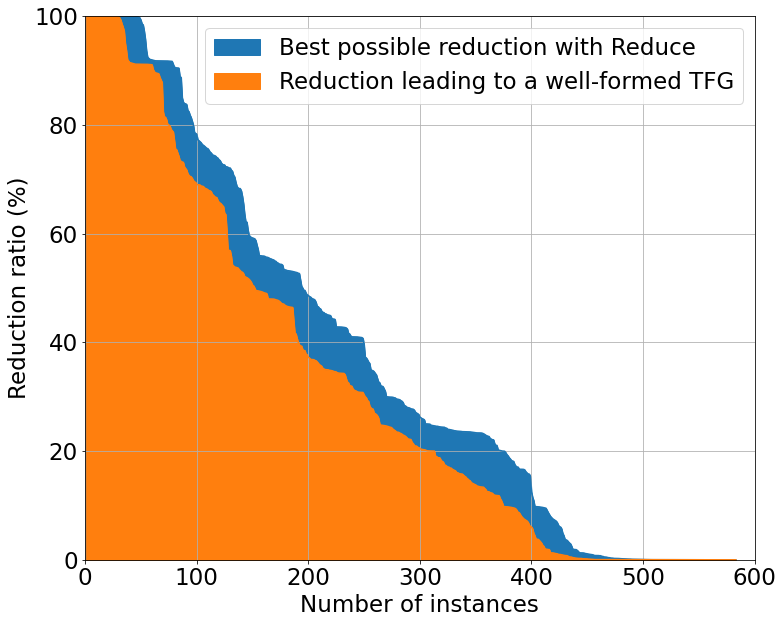}
	\caption{Safe instances}
	\label{fig:safe_reduction}
  \end{subfigure}
  \caption{Distribution of reduction ratios over: (a) all the instances in the MCC (b) all the safe instances in the MCC.}
	\label{fig:reduction} 
\end{figure}

We also observe that we lose few opportunities to reduce a net due to our
well-formedness constraint. Actually, we mostly lose the ability to simplify
some instances of ``partial'' marking graphs that could be reduced using
inhibitor arcs, or weights on the arcs (two features not supported by \caesar).

\subsubsection{Benchmark for marking reachability}

We evaluated the performance of \kong\ for the marking reachability
problem using a selection
of $426$ Petri nets taken from instances with a reduction ratio greater than
$1\%$. To avoid any bias introduced by models with a large number of instances,
we selected at most $5$ instances with a similar reduction ratio from each
model. For each instance, we generated $5$ queries that are markings
found using a ``random walk'' on the state space of the net (for this,
we used the tool \textsf{Walk} that is part of the Tina distribution). We ran \kong\  and \sift\ on
each of those queries with a time limit of \SI{5}{\minute}.

\subsubsection{Benchmark for concurrency relation}

We evaluated the performance of \kong\ on the $424$ instances of safe Petri nets
with a reduction ratio greater than $1\%$. We ran \kong\  and
\caesar\ on each of those instances, in two main modes: first with
a time limit of $\SI{15}{\minute}$ to compare the number of totally solved
instances (when the tool compute a complete concurrency matrix); next with a
timeout of $\SI{60}{\second}$ to compare the number of values (the filling
ratios) computed in the partial matrices.
Computation of a partial concurrency matrix with \caesar\ is done
in two phases: first a ``BDD exploration'' phase that can be stopped by the
user; then a post-processing phase that cannot be stopped. In practice this
means that the execution time on the initial net is often longer than with the
reduced one: the mean computation time for \caesar\ is about $\SI{48}{\second}$
and less than $\SI{12}{\second}$ for \kong.
In each test, we compared the output of \kong\  with the values obtained on the
initial net with \caesar, and achieved $100\%$ reliability.\\

Next, we give details about the results obtained with our experiments
and analyze the impact of using reductions.

\subsection{Results on marking reachability}

We display our results in the charts of
Fig.~\ref{fig:reachability_queries}, which compare the time needed to
compute a given number of queries, with and without using
reductions. (Note that we use a logarithmic scale for the time
value).  We consider two different samples of instances. First only
the instances with a high reduction ratio (in the interval $[0.5,
1]$), then the complete set of instances. 

We observe a clear advantage when we use reductions. For instance,
with instances that have a reduction ratio in the interval $[0.5, 1]$,
and with a time limit of \SI{5}{\minute}, we almost double the number
of computed queries (from $181$ with \sift\ alone, versus $357$ with
\kong). On the opposite, the small advantage of \sift\ alone, when the
running time is below $\SI{0.1}{\second}$, can be explained by the
fact that we integrate the running time of \reduce\ to the one of
\kong.

\begin{figure}[h]
	\centering
  \begin{subfigure}[b]{\linewidth}
    \includegraphics[width=\linewidth]{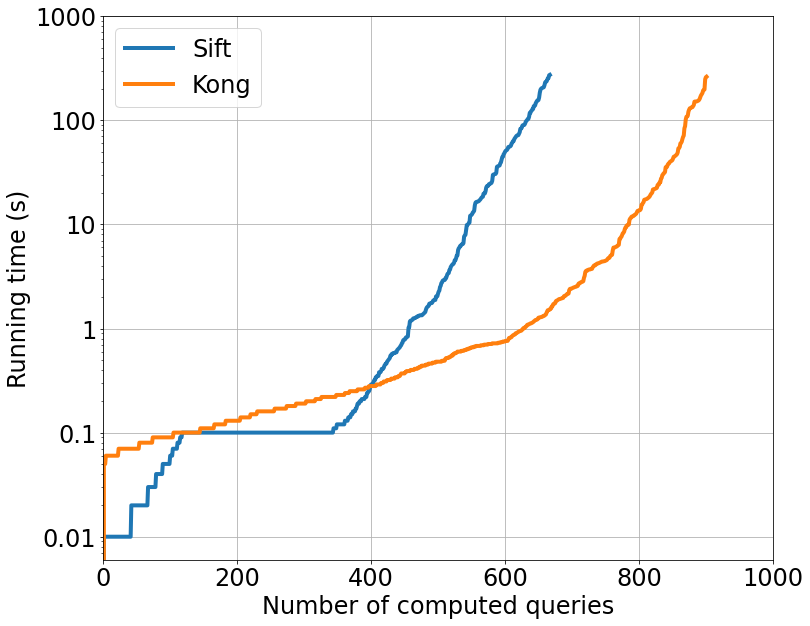}
    \caption{$r \in ]0, 1]$}
    \label{fig:reachability_queries_all} 
  \end{subfigure}
  
  \begin{subfigure}[b]{\linewidth}
    \includegraphics[width=\linewidth]{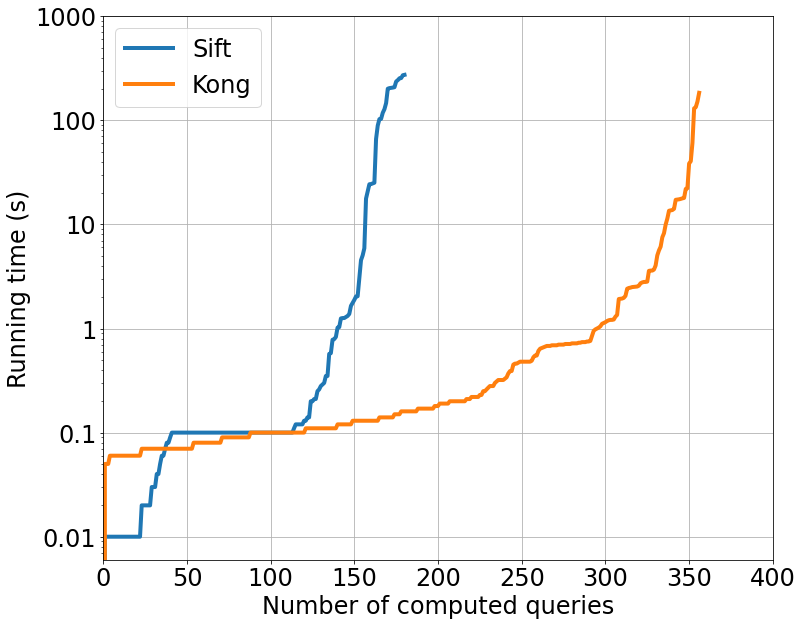}
    \caption{$r \in [0.5, 1]$}
    \label{fig:reachability_queries_high}
  \end{subfigure}

	\caption{Minimal running time to compute a given number of reachability
	queries for: (a) all instances, (b) instances with $r \in [0.5, 1]$.
	\label{fig:reachability_queries}}  
\end{figure}

\subsection{Totally computed concurrency matrices}

Our next results are for the computation of complete matrices, with a
timeout of $\SI{15}{\minute}$. We give the number of computed
instances in the table below. We split the results along three
different categories of instances, \emph{Low}/\emph{Fair}/\emph{High},
associated with different ratio ranges. We observe that we can compute
more results with reductions than without ($+30\%$). As could be
expected, the gain is greater on category \emph{High} ($+86\%$), but
it is still significant with the \emph{Fair} instances ($+28\%$).

\newcolumntype{x}[1]{>{\centering\arraybackslash\hspace{0pt}}p{#1}}
\begin{table*}[tb]
  \centering
  \begin{tabular}{l@{\quad}l c c c c x{1em} c c c}%
    \toprule
    \multicolumn{2}{c}{\multirow{2}{*}{\begin{minipage}[c]{6em}
          \centering\textsc{Reduction Ratio ($r$)} \end{minipage}}}
    && \multirow{2}{*}{
       \begin{minipage}[c]{6em}
         \centering \textsc{\# Test\\ Cases} \end{minipage}} &&
          \multicolumn{3}{c}{\textsc{\# Computed Matrices}}\\\cmidrule(rl){5-10} \multicolumn{2}{c}{} &&&&
           {\kong} && {\caesar}\\\midrule
    \emph{Low} & $r \in \; ]0, 0.25[$ && $160$ && \ratioc{85}{87} \\
    \emph{Fair} & $r \in \; [0.25, 0.5[$ && $112$ && \ratioc{49}{38} \\
    \emph{High} & $r \in \; [0.5, 1]$ && $152$ && \ratioc{95}{51} \\\hline\\[-1em]
    {Total} & $r \in \; ]0, 1]$ && $\fpeval{160 + 112 + 152}$ && \ratioc{\fpeval{85 + 49 + 95}}{\fpeval{87 + 38 + 51}}\\
    \bottomrule
  \end{tabular}
\end{table*}

Like in the previous case, we study the speed-up obtained with \kong\
using charts that compare the time needed to compute a given number of
instances; see Fig.~\ref{fig:complete_matrices}.

\begin{figure}[h]
	\centering
  \begin{subfigure}[b]{\linewidth}
    \includegraphics[width=\linewidth]{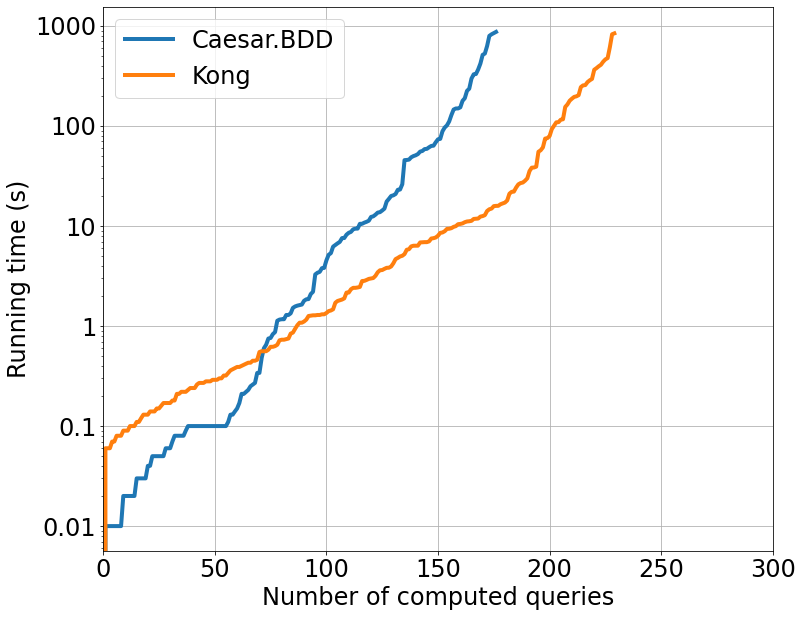}
    \caption{$r \in ]0, 1]$}
    \label{fig:complete_matrices_all} 
  \end{subfigure}
  
  \begin{subfigure}[b]{\linewidth}
    \includegraphics[width=\linewidth]{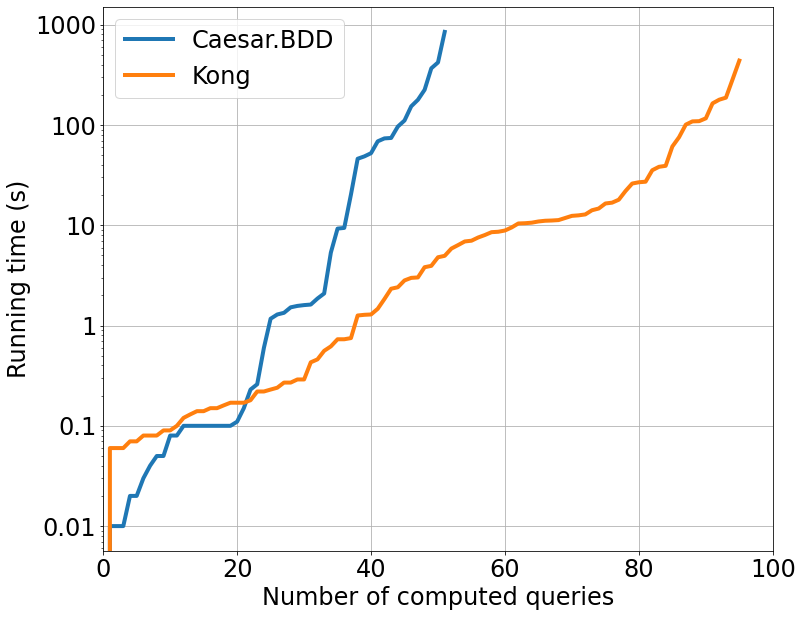}
    \caption{$r \in [0.5, 1]$}
    \label{fig:complete_matrices_high}
  \end{subfigure}

	\caption{Minimal running time to compute a given number of concurrency
	matrices for: (a) all instances, (b) instances with $r \in [0.5, 1]$.
	\label{fig:complete_matrices}}  
\end{figure}

\begin{figure}[tb]
  \centering
  \includegraphics[width=\linewidth]{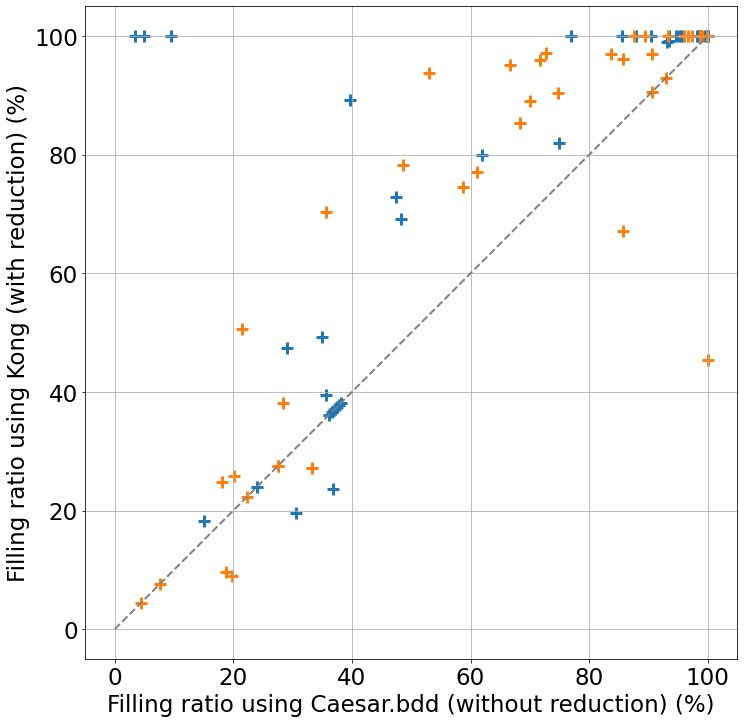}
  \caption{Comparing the filling ratio for partial matrices with \kong\
($y$-axis) and \caesar\ ($x$-axis) for instances with $r \in
[0.25, 0.5[$ (light orange) and $r \in [0.5, 1]$ (dark blue). (Computations done
with a BDD exploration timeout of \SI{60}{\second}.)}
	\label{fig:partial_computations}
\end{figure}

\subsection{Results with partial matrices}
We can also compare the ``accuracy'' of our approach when we have
incomplete results. To this end, we compute the concurrency relation
with a timeout of $\SI{60}{\second}$
on 
\caesar. We compare the \emph{filling ratio} obtained with and
without reductions. For a net with $n$ places, this ratio is given by the
formula ${2\, |\Conc|} / {(n^2 +n)}$, where $|\Conc|$ is the number of $0$s and
$1$s in the matrix.

We display our results using a scatter plot with linear scale, see
Fig.~\ref{fig:partial_computations}. We observe that almost all the
data points are on one side of the diagonal, meaning in this case that
reductions increase the number of computed values, with many examples
(top line of the plot) where we can compute the complete relation in
$\SI{60}{\second}$ only using reductions. The graphic does not
discriminate between the number of $1$s and $0$s, but we obtain
similar good results when we consider the filling ratio for only the
concurrent places (the $1$s) or only the nonconcurrent places (the
$0$s).

\newcommand{\partialm}[7]{${#1}\%$ & \num{#2} && \num{#3} & \num{#4} &
  $\times$\fpeval{round({#3} / {#4}, 1, 1)} &&
  $\times$\fpeval{round({#5}, 1, 1)} && \SI{\fpeval{round({#6}, 0, 1)}}{\second} &&
  \SI{\fpeval{round({#7}, 0, 1)}}{\second}}


\section{Conclusion and further work}

We propose a new method to transpose the computation of reachability
problems from an initial ``high-dimensionality'' domain (the set of
places in the initial net) into a smaller one (the set of places in
the reduced net).
Our approach is based on a combination of structural reductions with linear
equations first proposed in~\cite{berthomieu2018petri,berthomieu_counting_2019}.

Our main contribution, in the current work, is the definition of a new
data-structure that precisely captures the structure of these linear
equations, what we call the Token Flow Graph (TFG). We show how to use
the TFGs to accelerate the marking reachability problem, and also for
computing the concurrency relation, both in the complete and partial
cases.

We have several ideas on how to apply TFGs to other problems and how to extend
them. A natural application would be for model counting (our original goal
in~\cite{berthomieu2018petri}), where the TFG could lead to new algorithms for
counting the number of (integer) solutions in the systems of linear equations
that we manage.

Another possible application is the \emph{max-marking} problem, which means
finding the maximum of the expression $\sum_{p \in P} m(p)$ over all reachable
markings. On safe nets, this amounts to finding the maximal number of places
that can be marked together. We can easily adapt our algorithm to compute this
value and could even adapt it to compute the result when the net is not safe.

We can even manage a more general problem, related to the notion of
\emph{max-concurrent} sets of places. We say that the set $S$ is concurrent if
there is a reachable $m$ such that $m(p) > 0$ for all places $p$ in $S$. (This
subsumes the case of pairs and singleton of places.) The set $S$ is
\emph{max-concurrent} if no superset $S'\supsetneq S$ is concurrent.

Computing the max-concurrent sets of a net is interesting for several
reasons. First, it gives an alternative representation of the
concurrency relation that can sometimes be more space efficient: (1)
the max-concurrent sets provide a unique cover of the set of places of
a net, and (2) we have $p \conc q$ if and only if there is $S$
max-concurrent such that $\{p, q\} \subset S$. Obviously, on safe
nets, the size of the biggest max-concurrent set is the answer to the
\emph{max-marking} problem.

For future work, we would like to answer even more difficult questions, such as
proofs of Generalized Mutual Exclusion Constraints~\cite{giua1992generalized},
that requires checking invariants involving weighted sums over the marking of
places, of the form $\sum_{p \in P} w_p .  m(p)$. Another possible extension
will be to support non-ordinary nets (which would require adding weights on the
arcs of the TFG) and nets that are not safe (which can already be done with our
current approach, but require changing some ``axioms'' used in our
algorithm).\\

\noindent\textbf{Acknowledgements.} We would like to thank Pierre
Bouvier and Hubert Garavel for their insightful suggestions that
helped improve the quality of this paper.


\bibliographystyle{plain}
\balance
\bibliography{bibfile}

\end{document}